\documentclass[a4paper,UKenglish,numberwithinsect]{lipics-v2018}

\nolinenumbers

\usepackage{mathtools}
\usepackage{microtype}
\usepackage{amsthm}
\usepackage{amssymb}
\usepackage{complexity}

\usepackage{tikz}
\theoremstyle{plain}
\newtheorem{proposition}[theorem]{Proposition}

\makeatletter
\newcommand{\newreptheorem}[2]{\newtheorem*{rep@#1}{\rep@title}\newenvironment{rep#1}[1]{\def\rep@title{#2 \ref*{##1}}\begin{rep@#1}}{\end{rep@#1}}}
\makeatother
\newreptheorem{theorem}{Theorem}
\newreptheorem{proposition}{Proposition}

\newcommand{\ind}[1]{I_{#1}}
\newcommand{\sub}[1]{N_{#1}}
\newcommand{\ghom}[1]{\mathit{Hom}_{#1}}

\newcommand{\naut}[1]{\mathit{\#aut}(#1)}
\newcommand{\nsub}[2]{\mathit{\#sub}(#1, #2)}

\newcommand{\supg}{\sqsupseteq}
\newcommand{\psupg}{\sqsupset}
\newcommand{\subg}{\sqsubseteq}

\newcommand{\hommap}{\overset{\mathit{hom}}{\mapsto}}
\newcommand{\indmap}{\overset{\mathit{ind}}{\mapsto}}
\newcommand{\submap}{\overset{\mathit{sub}}{\mapsto}}
\newcommand{\redto}{\preceq}

\newcommand{\tw}[1]{\mathit{tw}(#1)}

\newcommand{\pdeg}[1]{\mathit{deg}(#1)}
\newcommand{\ml}[1]{\mathit{ML}(#1)}
\newcommand{\size}{\mathit{size}}

\newcommand{\npauthors}{Ne\v{s}et\v{r}il and Poljak}

\newcommand{\fpt}{\class{FPT}}
\newcommand{\fptw}[1]{\class{W[#1]}}

\title{Graph Pattern Polynomials}

\author{Markus Bl{\"a}ser}{Department of Computer Science, Saarland University, Saarland Informatics Campus, Saarbr{\"u}cken, Germany}{mblaeser@cs.uni-saarland.de}{https://www-cc.cs.uni-saarland.de/mblaeser/}{}

\author{Balagopal Komarath}{Saarland University, Saarland Informatics Campus, Saarbr{\"u}cken, Germany}{baluks@gmail.com}{http://www-cc.cs.uni-saarland.de/bkomarath/}{}

\author{Karteek Sreenivasaiah}{Department of Computer Science and Engineering, Indian Institute of
Technology Hyderabad, India}{karteek@iith.ac.in}{http://www.iith.ac.in/\textasciitilde karteek}{}

\authorrunning{M. Bl{\"a}ser, B. Komarath, K. Sreenivasaiah}
\Copyright{Markus Bl{\"a}ser, Balagopal Komarath, Karteek Sreenivasaiah}

\subjclass{, }  \hideLIPIcs

\begin{document}
\maketitle

\begin{abstract}
  Given a \emph{host} graph $G$ and a \emph{pattern} graph $H$, the
induced subgraph isomorphism problem is to decide whether $G$ contains
an induced subgraph that is isomorphic to $H$. We study the time
complexity of induced subgraph isomorphism problems where the pattern
graph is fixed. Ne\v{s}et\v{r}il and Poljak gave an $O(n^{k\omega})$
time algorithm that decides the induced subgraph isomophism problem
for \emph{any} $3k$ vertex pattern graph (the universal algorithm),
where $\omega(p, q, r)$ is the exponent of $n^p\times n^q$ and
$n^q\times n^r$ matrix multiplication and $\omega = \omega(1,1,1)$.

Algorithms that are faster than the universal algorithm are known only
for a finite number of pattern graphs. In this paper, we obtain
algorithms that are faster than the universal algorithm for infinitely
many pattern graphs. More specifically, we show that there exists a
family of pattern graphs $(H_{3k})_{k\geq 0}$ such that the induced
subgraph isomorphism problem for $H_{3k}$ (on $3k$ vertices) has a
$O(n^{\omega(k, k-1, k)}$ time algorithm when $k = 2^r, r \geq
1$. Note that for the currently known best matrix multiplication
algorithms $\omega(k, k-1, k) < k\omega$.

This algorithm is obtained by a reduction to the multilinear term
detection problem in a class of polynomials called \emph{graph pattern
polynomials}. We formally define this class of polynomials along with
a notion of reduction between these polynomials that allows us to
argue about the fine-grained complexity of isomosphism problems for
different pattern graphs. Besides the aforementioned result, we obtain
the following algorithms for induced subgraph isomorphism problems:
\begin{enumerate}
\item Faster than universal algorithm for $P_k$ ($k$-vertex paths)
  when $5\leq k\leq 9$ and $C_k$ ($k$-vertex cycles) for
  $k\in\{5, 7, 9\}$. In particular, we obtain $O(n^{\omega})$ time
  algorithms for $P_5$ and $C_5$ that are optimal under reasonable
  hardness assumptions.
\item Faster than universal algorithm for all pattern graphs except
  $K_k$ ($k$-vertex cliques) and $I_k$ ($k$-vertex independent sets)
  for $k\leq 8$.
\item Combinatorial algorithms (algorithms that do not use fast matrix
  multiplication) that take $O(n^{k-2})$ time for $P_k$ and $C_k$.
\item Combinatorial algorithms that take $O(n^{k-1})$ time for all
  pattern graphs except $K_k$ and $I_k$ for $k$.
\end{enumerate}

Our notion of reduction can also be used to argue about hardness of
detecting patterns within our framework. Since this method is used
(explicitly or implicitly) by many existing algorithms (including the
universal algorithm) for solving subgraph isomorphism problems, these
hardness results show the limitations of existing methods. We obtain
the following relative hardness results:
\begin{enumerate}
\item Induced subgraph isomorphism problem for any pattern containing
  a $k$-clique is at least as hard as $k$-clique.
\item For almost all patterns, induced subgraph isomorphism is harder
  than subgraph isomorphism.
\item For almost all patterns, the subgraph isomorphism problem for
  any of its supergraphs is harder than subgraph isomorphism for the
  pattern.
\end{enumerate}

\end{abstract}

\section{Introduction}

The \emph{induced subgraph isomorphism problem} asks, given simple and
undirected graphs $G$ and $H$, whether there is an induced subgraph of
$G$ that is isomorphic to $H$.  The graph $G$ is called the host graph
and the graph $H$ is called the pattern graph. This problem is
$\NP$-complete (See \cite{GareyJohnson90}, problem [GT21]). If the
pattern graph $H$ is fixed, there is a simple $O(n^{|V(H)|})$ time
algorithm to decide the induced subgraph isomorphism problem for
$H$. We study the time complexity of the induced subgraph isomorphism
problem for fixed pattern graphs on the Word-RAM model.

The earliest non-trivial algorithm for this problem was given by Itai
and Rodeh \cite{IR78} who showed that the the number of triangles can
be computed in $O(n^\omega)$ time on $n$-vertex graphs, where $\omega$
is the exponent of matrix multiplication. Later,
\npauthors\ \cite{NP85} generalized this algorithm to count $K_{3k}$
in $O(n^{k\omega})$ time, where $K_{3k}$ is the clique on $3k$
vertices. Eisenbrand and Grandoni \cite{Eisenbrand2004} extended this
algorithm further to count $K_{3k+j}$ for $j\in\{0, 1, 2\}$ using
rectangular matrix multiplication in $O(n^{\omega(k + \lceil j/2
  \rceil, k, k + \lfloor j/2 \rfloor)})$ time. This algorithm uses
fast matrix multiplication to achieve the speedup and in fact works
for all pattern graphs on $3k+j$ vertices. Hence we call this
algorithm the universal algorithm. It is reasonable to expect that one
might be able to obtain faster algorithms for specific pattern
graphs. However, algorithms faster than the universal algorithm are
only known for finitely many pattern graphs.

Algorithms that do not use fast matrix multiplication, called
\emph{combinatorial algorithms}, have also been studied. No
combinatorial algorithm that beats the trivial $O(n^k)$ time algorithm
is known for detecting $k$-cliques in $n$ vertex graphs. However,
improvements for certain pattern graphs such as $K_k-e$ has been shown
by Virginia Williams (See \cite{VV08}, p.45). They show a
combinatorial algorithm that decides the induced subgraph isomorphism
problem for $K_k - e$ in time $O(n^{k-1})$. An $O(n^{k-1})$
combinatorial algorithm is also known for deciding induced subgraph
isomorphism problem for $P_k$.

The use of algebraic methods has been particularly useful in finding
fast combinatorial algorithms for detecting pattern graphs. Ryan
Williams \cite{Williams09} gave a linear time algorithm for the (not
necessarily induced) subgraph isomorphism problem for $P_k$. This was
later generalized by Fomin, Lokshtanov, Raman, Saurabh, and Rao
 \cite{FominLRSR12} to give $O(n^{\tw{H}+1})$ time algorithms for the
(not necessarily induced) subgraph isomorphism problem for $H$ in $n$
vertex graphs.  These results use efficient constructions for
\emph{homomorphism polynomials} (defined later).

The question of whether improving algorithms for detecting a certain
pattern implies faster algorithms for another pattern has also been
studied. In particular, \npauthors\ show that improved algorithms for
detecting $k$-cliques yield improved algorithms for all $k$-vertex
pattern graphs. More precisely:

\begin{theorem}{(\cite{NP85})}
  If the induced subgraph isomorphism problem for $K_k$ can be decided
  in $O(n^{f(k)})$ time for some $f(k)$, then the induced subgraph
  isomorphism problem for $H$ can be decided in time $O(n^{f(k)})$
  time, where $H$ is any $k$-vertex pattern graph.
  \label{thm:np-universal-full}
\end{theorem}

In this sense, the $k$-clique is a \emph{universal} pattern.

\npauthors's \cite{NP85} algorithm can be easily modified to output the
homomorphism polynomial no host graphs of $n$ vertices for the pattern
$K_{3k}$ in $O(n^{k\omega})$ time given $1^n$ as input. For cliques,
counting (or detecting) homomorphisms\footnote{For host $G$ and
  pattern $H$, a function $f:V(H) \mapsto V(G)$ such that
  $\{u, v\}\in E(H) \implies \{f(u), f(v)\}\in E(G)$} and counting
(or detecting) induced subgraph isomorphisms have the same
complexity. It is unclear whether computing homomorphism polynomials
efficiently for other pattern graphs help with the induced subgraph
isomorphism problem for those pattern graphs.

\subsection*{Our Results}
In this paper, we show that we can obtain algorithms that are faster
than the universal algorithm for infinitely many pattern graphs.

\begin{reptheorem}{thm:H3k-full}
  There exists a family of pattern graphs ${(H_{3k})}_{k\geq 0}$ where
  $H_{3k}$ is a $3k$-vertex graph such that the induced subgraph
  isomorphism problem for $H_{3k}$ has an $O(n^{\omega(k, k-1, k)})$
  time algorithm for infinitely many $k$.
\end{reptheorem}

Here, $\omega(p, q, r)$ is the exponent of $n^p\times n^q$ and
$n^q\times n^r$ matrix multiplication. The exponent of matrix
multiplication is defined as $\omega = \omega(1, 1, 1)$ (See
\cite{DBLP:journals/toc/Blaser13} for a more detailed
introduction). The best known algorithm for $K_{3k}$ takes time
$O(n^{k\omega})$ and the upper-bound on $\omega(k, k-1, k)$ is
strictly smaller than the upper-bound on $k\omega$ for the currently
known fastest matrix multiplication algorithms. If $\omega = 2$, then
we have $\omega(k, k-1, k) = k\omega$ and the above algorithm does not
improve upon the universal algorithm. However, the best upper-bound
that we know is $\omega < 2.373$. It is known that current methods
cannot show $\omega = 2$ (See \cite{AmbainisFG15}, \cite{AlmanW18}).

We develop an algebraic framework to study algorithms for the induced
subgraph isomorphism problems where we consider the size of the
pattern graph a constant. The above algorithm is obtained using this
framework. We show that the existing algorithms for natural pattern
graphs such as $k$-paths and $k$-cycles can be improved by efficiently
computing homomorphism polynomials for pattern graphs that are much
sparser than $k$-cliques.

We obtain, in Theorem~\ref{thm:AlgosForPaths-full} and
Theorem~\ref{thm:AlgosForCycles-full}, the following faster
(randomized, one-sided error) algorithms:

\begin{itemize}
\item Faster algorithms for induced subgraph isomorphism problem for
  $P_k$ for $5\leq k \leq 9$.
\item Faster algorithms for induced subgraph isomorphism problem for
  $C_k$ for $k\in\{5,7,9\}$.
\item $O(n^{k-2})$ time combinatorial algorithm for induced subgraph
  isomorphism problem for $P_k$ and $C_k$.
\item $O(n^{k-2})$ time deterministic combinatorial algorithms for
  computing the parity of the number of induced subgraphs isomorphic
  to $P_k$ and $C_k$ in $n$-vertex graphs.
\end{itemize}

Unfortunately, we do not know how to compute these homomorphism
polynomials for smaller graphs using circuits of size smaller than
that for homomorphism polynomials for $k$-cliques when $k$ is
arbitrary. Therefore, we do not have an improvement similar to the one
in Theorem~\ref{thm:H3k-full} for paths or cycles.

In light of Theorem~\ref{thm:np-universal-full}, which shows that
$k$-cliques are universal, we show that homomorphism polynomials for
$K_k - e$, the $k$-vertex graph obtained by deleting an edge from
$K_k$, are \emph{almost} universal. We show that the arithmetic
circuit complexity of $\ghom{K_k - e}$ can be used to unify many
existing results. We show that if $\ghom{K_k - e}$ has $O(n^{f(k)})$
size circuits for some function $f(k)$, then:

\begin{enumerate}
\item (Theorem~\ref{thm:Kk-e-full}) The induced subgraph isomorphism
  problem for all $k$-vertex pattern graphs other than $K_k$ and $I_k$
  can be decided by an $O(n^{f(k)})$ time algorithm, where $k$ is
  regarded as a constant and $f(k)$ is any function of
  $k$. (\cite{VV08} gives a combinatorial algorithm for $K_k - e$,
  \cite{FKLLTCS15} gives an algorithm for $P_k$)

\item (Theorem~\ref{thm:Kk-e-counting-full}) If there is an $O(t(n))$ time
  algorithm for counting the number of induced subgraph isomorphisms
  for a $k$-vertex pattern $H$, then the number of induced subgraph
  isomorphisms for all $k$-vertex patterns can be computed in
  $O(n^{f(k)} + t(n))$ time on $n$-vertex graphs. (\cite{KloksKM00}
  gives this result for $k = 4$ and \cite{KowalukLL13} gives a weaker
  result similar to this one)
\end{enumerate}

Even though these theorems are known for specific values of $f(k)$ as
cited above. We believe that the connection to homomorphism
polynomials for $\ghom{K_k - e}$ is new.

On the lower bounds front, we show in
Theorem~\ref{thm:supergraphs-full}, Theorem~\ref{thm:clique-hard-full}
and Theorem~\ref{thm:inducedHarderThanNonInduced-full} that within the
framework that we develop:
  \begin{enumerate}
  \item The induced subgraph isomorphism problem for any pattern
    containing a $k$-clique or a $k$-independent set is at least as
    hard as the isomorphism problem for $k$-clique.
  \item For almost all pattern graphs $H$, the induced subgraph
    isomorphism problem for $H$ is harder than the subgraph
    isomorphism problem for $H$.
  \item For almost all pattern graphs $H$, the subgraph isomorphism
    problem for $H$ is easier than subgraph isomorphism problems for
    all supergraphs of $H$.
  \end{enumerate}

We note that only randomized algorithmic reductions are known for
Part~2 of the above theorem and Part~3 is unknown. It is not clear
whether our reductions imply algorithmic hardness for these problems.

\subsection*{Technique}
The \emph{Homomorphism polynomial} for a pattern graph $H$ denoted
$\ghom{H,n}$ is a polynomial such that the monomials of the polynomial
correspond one-to-one with homomorphisms from $H$ to an $n$-vertex
graph. Similarly, we define the graph pattern polynomial families
$\ind{H} = {(\ind{H,n})}_{n\geq 0}$ and $\sub{H} =
{(\sub{H,n})}_{n\geq 0}$ that correspond to the induced subgraph
isomorphism problem for $H$ and the (not necessarily induced) subgraph
isomorphism problem\footnote{Given $(G, H)$, decide whether there
  exists an injective $f:V(H) \mapsto V(G)$ such that $\{u, v\}\in
  E(H) \implies \{f(u), f(v)\}\in E(G)\}$.} for $H$ respectively.  It
can be shown that testing for subgraph isomorphism is equivalent to
testing whether the homomorphism polynomial has multilinear terms
because subgraph isomorphisms are exactly the injective
homomorphisms. In fact, any polynomial family $f$ such that the
multilinear terms of $f$ correspond to multilinear terms of $\sub{H}$
is enough. This naturally leads to a notion of reduction between these
graph pattern polynomial families (denoted by $\redto$). For example,
we say that $\sub{H}\redto\ghom{H}$). This notion of reduction allows
us to compare the hardness of different pattern detection problems as
well as construct new algorithms as follows:

\begin{repproposition}{prop:red-alg-full}
  Let $f$ and $g$ be graph pattern polynomial families. If $f\redto g$
  and $g$ has $O(n^{s(k)})$ size arithmetic circuits, then we can detect
  patterns corresponding to $f$ using an $O(n^{s(k)})$ time algorithm.
\end{repproposition}

This framework naturally raises the question whether one can find
families $f$ such that $\sub{H}\redto f$ and $f$ has smaller circuits
than $\ghom{H}$. We show that this is not possible by showing that in
this case $\ghom{H}$ has circuits that is as small as circuits for
$f$.

\subsection*{Other related work}
Curticapean, Dell, and Marx \cite{CDM17} showed that algorithms that
count homomorphisms can be used to count subgraph isomorphisms.
Williams, Wang, Williams, and Yu \cite{WWWY15} gave $O(n^\omega)$ time
algorithms for the induced subgraph isomorphism problems for four
vertex pattern graphs, except for $I_4$ and $K_4$. Floderus, Kowaluk,
Lingas, and Lundell \cite{FKLLTCS15} invented a framework that gives
$O(n^{k-1})$ combinatorial algorithms for induced subgraph isomorphism
problems for many pattern graphs on $k$ vertices.

Floderus, Kowaluk, Lingas, and Lundell \cite{FKLLSIAM15} showed
reductions between various induced subgraph isomorphism problems.
They proved that induced subgraph isomorphism problem for $H$ when $H$
contains a $k$-clique (or $k$-independent set) that is vertex-disjoint
from all other $k$-cliques (or $k$-independent sets) is at least as
hard as the induced subgraph isomorphism problem for $K_k$. They also
proved that detecting an induced $C_4$ is at least as hard as
detecting a $K_3$. This is the only example known where a pattern is
harder than another pattern that is not a subgraph. Hardness results
are also known for arithmetic circuits computing homomorphism
polynomials. Austrin, Kaski, and Kubjas \cite{AKK17} proved that
tensor networks (a restricted form of arithmetic circuits) computing
homomorphism polynomials for $k$-cliques require $\Omega(n^{\lceil 2k
  / 3 \rceil})$ time. Durand, Mahajan, Malod, Rugy-Altherre, and
Saurabh \cite{DMMRS14} proved that homomorphism polynomials for
certain pattern families are complete for the class $\VP$, the
algebraic analogue of the class $\P$. This is the only known
polynomial family that is complete for $\VP$ other than the canonical
complete family of universal circuits.

\section{Preliminaries}

For a polynomial $f$, we use $\pdeg{f}$ to denote the degree of $f$. A
monomial is called multilinear, if every variable in it has degree at
most one.  We use $\ml{f}$ to denote the multilinear part of $f$, that
is, the sum of all multilinear monomials in $f$. An arithmetic circuit
computing a polynomial $P\in K[x_1,\dotsc,x_n]$ is a circuit with $+$,
$\times$ gates where the input gates are labelled by variables or
constants from the underlying field and one gate is designated as the
output gate. The size of an arithmetic circuit is the number of wires
in the circuit. For indeterminates $x_1,\dots,x_n$ and a set
$S = \{s_1, \dotsc, s_p\} \subseteq \{1,\dots,n\}$ of indices, we
write $x_S$ to denote the product $x_{s_1}\dotsm x_{s_p}$.

An induced subgraph isomorphism from $H$ to $G$ is an injective
function $\phi : V(H) \indmap V(G)$ such that
$\{u, v\}\in E(H) \iff \{\phi(u), \phi(v)\}\in E(G)$.  Any function
from $V(H)$ to $V(G)$ can be extended to unordered pairs of vertices
of $H$ as $\phi(\{u, v\}) = \{\phi(u), \phi(v)\}$. A subgraph
isomorphism from $H$ to $G$ is an injective function
$\phi : V(H) \submap V(G)$ such that
$\{u, v\}\in E(H) \implies \{\phi(u), \phi(v)\}\in E(G)$. Two subgraph
isomorphisms or induced subgraph isomorphisms are considered different
only if the set of edges in the image are different. A graph
homomorphism from $H$ to $G$ is a function $\phi: V(H) \hommap V(G)$
such that $\{u, v\}\in E(H) \implies \{\phi(u), \phi(v)\}\in
E(G)$. Unlike isomorphisms, we consider two distinct functions that
yield the same set of edges in the image as distinct graph
homomorphisms. We define $\phi(S) = \{\phi(s) : s\in S\}$.

We write $H\subg H'$ ($H\supg H'$) to specify that $H$ is a subgraph
(supergraph) of $H'$. The number $\tw{H}$ stands for the treewidth of
$H$. The graph $K_n$ is the complete graph on $n$ vertices labelled
using $[n]$. We use the fact that $\naut{H} = 1$ for almost all graphs
in many of our results. In this paper, we will frequently consider
graphs where vertices are labelled by tuples. A vertex $(i, p)$ is
said to have \emph{label} $i$ and \emph{colour} $p$. An edge $\{(i_1,
p_1), (i_2, p_2)\}$ has \emph{label} $\{i_1, i_2\}$ and \emph{colour}
$\{p_1, p_2\}$. We will sometimes write this edge as $(\{i_1, i_2\},
\{p_1, p_2\})$. Note that both $\{(i_1, p_1), (i_2, p_2)\}$ and
$\{(i_2, p_1), (i_1, p_2)\}$ are written as $(\{i_1, i_2\}, \{p_1,
p_2\})$. But the context should make it clear which edge is being
rewritten.

\section{A Motivating Example: Induced-$P_4$ Isomorphism}

In this section, we sketch a one-sided error, randomized $O(n^2)$ time
algorithm for the induced subgraph isomorphism problem for $P_4$ to
illustrate the techniques used to derive algorithms in this paper.

We start by giving an algorithm for the subgraph isomorphism problem
for $P_4$. Consider the following polynomial:

\begin{equation*}
  \sub{P_4,n} = \sum_{(p, q, r, s) : p < s} y_p y_q y_r y_s x_{\{p, q\}} x_{\{q, r\}} x_{\{r, s\}}
\end{equation*}

where the summation is over all quadruples over $[n]$ where all four
elements are distinct. Each monomial in the above polynomial
corresponds naturally to a $P_4$ in an $n$-vertex graph. The condition
$p < s$ ensures that each path has exactly one monomial corresponding
to it.

Given an $n$-vertex host graph $G$ and an arithmetic circuit for
$\sub{P_4, n}$, we can construct an arithmetic circuit for the
polynomial $\sub{P_4, n}(G)$ on the $y$ variables obtained by
substituting $x_e = 0$ when $e\not\in E(G)$ and $x_e = 1$ when
$e\in E(G)$. The polynomial $\sub{P_4, n}(G)$ can be written as
$\sum_{X} a_{X} y_{X}$ where the summation is over all four vertex
subsets $X$ of $V(G)$ and $a_X$ is the number of $P_4$s in the induced
subgraph $G[X]$. Therefore, we can decide whether $G$ has a subgraph
isomorphic to $P_4$ by testing whether $\sub{P_4, n}(G)$ is
identically $0$. Since the degree of this polynomial is a constant
$k$, this can be done in time linear in the size of the arithmetic
circuit computing $\sub{P_4, n}$.

However, we do not know how to construct a $O(n^2)$ size arithmetic
circuit for $\sub{P_4, n}$. Instead, we construct a $O(n^2)$ size
arithmetic circuit for the following polynomial called the walk
polynomial:

\begin{equation*}
  \ghom{P_4,n} = \sum_{\phi: P_4 \hommap K_n} \prod_{v\in V(P_4)}
  z_{v, \phi(v)} y_{\phi(v)} \prod_{e\in E(P_4)} x_{\phi(e)}
\end{equation*}

This polynomial is also called the homomorphism polynomial for $P_4$
because its terms are in one-to-one correspondence with graph
homomorphisms from $P_4$ to $K_n$.  As before, we consider the
polynomial $\ghom{P_4, n}(G)$ obtained by substituting for the $x$
variables appropriately. The crucial observation is that
$\ghom{P_4, n}(G)$ contains a multilinear term if and only if
$\sub{P_4, n}(G)$ is not identically zero. This is because the
multilinear terms of $\ghom{P_4, n}$ correspond to injective
homomorphisms from $P_4$ which in turn correspond to subgraph
isomorphisms from $P_4$. More specifically, each $P_4$ corresponds to
two injective homomorphisms from $P_4$ since $P_4$ has two
automorphisms. Therefore, we can test whether $G$ has a subgraph
isomorphic to $P_4$ by testing whether $\ghom{P_4, n}(G)$ has a
multilinear term. We can construct a $O(n^2)$ size arithmetic circuit
for the polynomial $p_4 = \ghom{P_4, n}$ inductively as follows:

\begin{align*}
  p_{1, v} &= y_v, v\in [n]\\
  p_{i+1, v} &= \sum_{u\in [n]} p_{i, u} y_v x_{\{u, v\}}, v\in [n], i\geq 1\\
  p_4 &= \sum_{v\in [n]} p_{4, v}
\end{align*}

The above construction works for any $k$ and not just $k = 4$. This
method is used by Ryan Williams \cite{Williams09} to obtain an
$O(2^{k}(n+m))$ time algorithm for the subgraph isomorphism problem
for $P_k$.

In fact, the above method works for any pattern graph $H$. Extend the
definitions above to define $\sub{H, n}$ and $\ghom{H, n}$ in the
natural fashion. Then, we can test whether an $n$-vertex graph $G$ has
a subgraph isomorphic to $H$ by testing whether $\sub{H, n}(G)$ is
identically zero which in turn can be done by testing whether
$\ghom{H, n}(G)$ has a multilinear term. Therefore, the complexity of
subgraph isomorphism problem for any pattern $H$ is as easy as
constructing the homomophism polynomial for $H$. This method is used
by Fomin et. al. \cite{FominLRSR12} to obtain efficient algorithms for
subgraph isomorphism problems.

We now turn our attention to the induced subgraph isomorphism problem
for $P_4$. We note that the induced subgraph isomorphism problem for
$P_k$ is much harder than the subgraph isomorphism problem for
$P_k$. The subgraph isomorphism problem for $P_k$ has a linear time
algorithm as seen above but the induced subgraph isomorphism problem
for $P_k$ cannot have $n^{o(k)}$ time algorithms unless
$\fpt = \fptw1$. We start by considering the polynomial:

\begin{equation*}
  \ind{P_4,n} = \sum_{(p, q, r, s) : p < s} y_p y_q y_r y_s  x_{\{p, q\}}  x_{\{q, r\}} x_{\{r, s\}} (1 - x_{\{p, r\}})(1 - x_{\{p, s\}})(1 - x_{\{q, s\}})
\end{equation*}

The polynomial $\ind{P_4, n}(G)$ can be written as $\sum_X y_X$ where
the summation is over all four vertex subsets of $V(G)$ that induces a
$P_4$. Notice that all coefficents are $1$ because there can be at
most $1$ induced-$P_4$ on any four vertex subset. By expanding terms
of the form $1-x_*$ in the above polynomial, we observe that we can
rewrite $\ind{P_4, n}$ as follows:

\begin{equation*}
  \ind{P_4, n} = \sub{P_4, n} - 4\sub{C_4, n} - 2\sub{K_3+e, n} +
  6\sub{K_4-e, n} + 12\sub{K_4, n}
\end{equation*}

Since the coefficients in $\ind{P_4, n}(G)$ are all $0$ or $1$, it is
sufficient to check whether $\ind{P_4, n}(G) \pmod 2$ is non-zero to
test whether $\ind{P_4, n}(G)$ is non-zero. From the above equation,
we can see that $\ind{P_4, n} = \sub{P_4, n} \pmod 2$. Therefore,
instead of working with $\ind{P_4, n} \pmod 2$, we can work with
$\sub{P_4, n} \pmod 2$. We have already seen that we can use
$\ghom{P_4, n}(G)$ to test whether $\sub{P_4, n}(G)$ is
non-zero. However, this is not sufficient to solve induced subgraph
isomorphism. We want to detect whether $\sub{P_4, n}(G)$ is non-zero
modulo $2$. Therefore, the multilinear terms of $\ghom{P_4, n}(G)$ has
to be in one-to-one correspondence with the terms of
$\sub{P_4, n}(G)$. We have to divide the polynomial $\ghom{P_4, n}(G)$
by $2$ before testing for the existence of multilinear terms modulo
$2$. However, since we are working over a field of characteristic $2$,
this division is not possible. We work around this problem by starting
with $\ghom{P_4, n'}$ for $n'$ slightly larger than $n$ and we show
that this enables the ``division'' by $2$.

The reader may have observed that instead of the homomorphism
polynomial, we could have taken any polynomial $f$ for which the
multilinear terms of $f(G)$ are in one-to-one correspondence with
$\sub{P_4, n}(G)$. This observation leads to the definition of a
notion of reduction between polynomials. Informally, $f\redto g$ if
detecting multilinear terms in $f(G)$ is as easy as detecting
multilinear terms in $g(G)$. Additionally, for the evaluation $f(G)$
to be well-defined, the polynomial $f$ must have some special
structure. We call such polynomials graph pattern polynomials.

On first glance, it appears hard to generalize this algorithm for
$P_4$ to sparse pattern graphs on an arbitrary number of vertices (For
example, $P_k$) because we have to argue about the coefficients of
many $\sub{*}$ polynomials in the expansion. On the other hand, if we
consider the pattern graph $K_k$, we have $\ind{K_k} = \ghom{K_k}$. In
this paper, we show that for many graph patterns sparser than $K_k$,
the induced subgraph isomorphism problem is as easy as constructing
arithmetic circuits for homomorphism polynomials for those patterns
(or patterns that are only slightly denser).

\section{Graph pattern polynomial families}

We will consider polynomial families $f = (f_n)$ of the following
form: Each $f_n$ will be a polynomial in variables $y_1,\dots,y_n$,
the vertex variables, and variables $x_1,\dots,x_{\binom n 2}$, the
edge variables, and at most linear in $n$ number of additional
variables.The degree of each $f_n$ will usually be constant.

The (not necessarily induced) subgraph isomorphism polynomial family
$\sub H = {(\sub{H,n})}_{n\geq 0}$ for a fixed pattern graph $H$ on
$k$ vertices and $\ell$ edges is a family of multilinear polynomials
of degree $k+\ell$. The $n^\text{th}$ polynomial in the family,
defined over the vertex set $[n]$, is the polynomial on
$n + \binom{n}{2}$ variables given by (\ref{eq:sub-full}):

\begin{equation}
  \sub{H,n} = \sum_{\phi : V(H)\submap V(K_n)} y_{\phi(V(H))} x_{\phi(E(H))} 
  \label{eq:sub-full}
\end{equation} 

When context is clear, we will often omit the subscript $n$ and simply
write $\sub H$. Given a (host) graph $G$ on $n$ vertices, we can
substitute values for the edge variables of $\sub{H,n}$ depending on
the edges of $G$ ($x_e = 1$ if $e\in E(G)$ and $x_e = 0$ otherwise) to
obtain a polynomial $\sub{H,n}(G)$ on the vertex variables. The
monomials of this polynomial are in one-to-one correspondence with the
$H$-subgraphs of $G$. i.e., a term $ay_{v_1}\dotsm y_{v_k}$, where $a$
is a positive integer, indicates that there are $a$ subgraphs
isomorphic to $H$ in $G$ on the vertices $v_1,\dotsc,v_k$. Therefore,
to detect if there is an $H$-subgraph in $G$, we only have to test
whether $\sub{H,n}(G)$ has a multilinear term.

The induced subgraph isomorphism polynomial family
$\ind{H} = {(\ind{H,n})}_{n\geq 0}$ for a pattern graph $H$ over the
vertex set $[n]$ is defined in (\ref{eq:ind-full}).

\begin{equation}
  \ind{H,n} = \sum_{\phi : V(H)\indmap V(K_n)} y_{\phi(V(H))} x_{\phi(E(H))} \prod_{e\not\in E(H)} (1 - x_{\phi(e)})
  \label{eq:ind-full}
\end{equation}

If we substitute the edge variables of $\ind{H,n}$ using a host graph
$G$ on $n$ vertices, then the monomials of the resulting polynomial
$\ind{H,n}(G)$ on the vertex variables are in one-to-one
correspondence with the induced $H$-subgraphs of $G$. In particular,
all monomials have coefficient $0$ or $1$ because there can be at most
one induced copy of $H$ on a set of $k$ vertices. This implies that to
test if there is an induced $H$-subgraph in $G$, we only have to test
whether $\ind{H,n}(G)$ has a multilinear term and we can even do this
modulo $p$ for any prime $p$. Also, note that $\ind{H}$ is simply
$\ind{\overline{H}}$ where all the edge variables $x_e$ are replaced
by $1 - x_e$.

The homomorphism polynomial family $\ghom{H} = {(\ghom{H,n})}_{n\geq
  0}$ for pattern graph $H$ over the vertex set $[n]$ is defined in
(\ref{eq:hom-full}).

\begin{equation}
  \ghom{H,n} = \sum_{\phi: H \hommap K_n}
  \prod_{v\in V(H)} z_{v, \phi(v)} y_{\phi(v)} \prod_{e\in E(H)} x_{\phi(e)}
  \label{eq:hom-full}
\end{equation}

The variables $z_{a, v}$'s are called the \emph{homomorphism
  variables}.  They keep track how the vertices of $H$ are mapped by
the different homomorphisms in the summation. We note that the size of
the arithmetic circuit computing $\ghom{H,n}$ is independent of the
labelling chosen to define the homomorphism polynomial.

The induced subgraph isomorphism polynomial for any graph $H$ and
subgraph isomorphism polynomials for supergraphs of $H$ are related as
follows:

\begin{equation}\label{eq:indsub-full}
  \ind{H,n} = \sum_{H' \supg H} {(-1)}^{e(H') - e(H)} \nsub{H}{H'} \sub{H',n}
\end{equation}

Here $e(H)$ is the number of edges in $H$ and $\nsub{H}{H'}$ is the
number of times $H$ appears as a subgraph in $H'$. The sum is taken
over all supergraphs $H'$ of $H$ having the same vertex set as
$H$. Equation~\ref{eq:indsub-full} is used by Curticapean, Dell, and
Marx \cite{CDM17} in the context of counting subgraph isomorphisms.

\begin{example}
  Let $P_3$ be the path on $3$ vertices and let $K_3$ be the triangle.
  \begin{align*}
    \sub{P_3, 3} &= y_1y_2y_3(x_{\{1,2\}}x_{\{2,3\}} +
                   x_{\{1,3\}}x_{\{2,3\}} +
                   x_{\{1,2\}}x_{\{1,3\}})\\
    \ind{P_3, 3} &= y_1y_2y_3\bigl(x_{\{1,2\}}x_{\{2,3\}}(1 - x_{\{1,3\}})\\
                 &+ x_{\{1,3\}}x_{\{2,3\}}(1 - x_{\{1,2\}})\\
                 &+ x_{\{1,2\}}x_{\{1,3\}}(1 - x_{\{2,3\}})\bigr)\\
                 &= \sub{P_3,3} - 3\sub{K_3,3}
  \end{align*}
\end{example}

For any fixed pattern graph $H$, the degree of polynomial families
$\sub{H}$, $\ind{H}$, and $\ghom{H}$ are bounded by a constant
depending only on the size of $H$. Such polynomial families are called
constant-degree polynomial families.

\begin{definition}
  A constant-degree polynomial family $f = (f_n)$ is called a
  \emph{graph pattern} polynomial family if the $n^{\text{th}}$
  polynomial in the family has $n$ vertex variables, $\binom{n}{2}$
  edge variables, and at most $cn$ other variables, where $c$ is a
  constant, and every non-multilinear term of $f_n$ has at least one
  non-edge variable of degree greater than 1.
\end{definition}


It is easy to verify that $\ind{H}$, $\sub{H}$, and $\ghom{H}$ are all
graph pattern polynomial families. For a graph pattern polynomial $f$,
we denote by $f(G)$ the polynomial obtained by substituting $x_e = 0$
if $e\not\in E(G)$ and $x_e = 1$ if $e\in E(G)$ for all edge variables
$x_e$. Note that for any graph pattern polynomial $f$, we have
$\ml{f(G)} = \ml{f}(G)$. This is because any non-multilinear term in
$f$ has to remain non-multilinear or become 0 after this substitution.

\begin{definition}
  \begin{enumerate}
  \item A constant degree polynomial family $f = (f_n)$ has circuits
    of size $s(n)$ if there is a sequence of arithmetic circuits
    $(C_n)$ such that $C_n$ computes $f_n$ and has size at most
    $s(n)$.
  \item $f$ has uniform $s(n)$-size circuits, if on input $n$, we can
    construct $C_n$ in time $O(s(n))$ on a Word-RAM.\footnote{Since we
      are dealing with fine-grained complexity, we have to be precise
      with the encoding of the circuit.  We assume an encoding such
      that evaluating the circuit is linear time and substituting for
      variables with polynomials represented by circuits is
      constant-time.}
  \end{enumerate}
\end{definition}

We now define a notion of reducibility among graph pattern
polynomials.

\begin{definition} \label{def:subfam-full}
A \emph{substitution family} $\sigma = (\sigma_n)$ 
is a family of mappings 
\[
  \sigma_n: \{y_1,\dots,y_n,x_1,\dots,x_{\binom n2},u_1,\dots,u_{m(n)}\}
  \to K[y_1,\dots,y_{n'},x_1,\dots,x_{\binom {n'}2},v_1,\dots,v_{r(n)}]
\]
mapping variables to polynomials such that:
\begin{enumerate}
  \item $\sigma$ maps vertex variables to constant-degree monomials
    containing one or more vertex variables or other variables, and
    no edge variables.
  \item $\sigma$ maps edge variables to polynomials with
    constant-size circuits containing at most one edge variable and
    no vertex variables.
  \item $\sigma$ maps other variables to constant-degree monomials
    containing no vertex or edge variables and at least one
    other variable.
\end{enumerate}
$\sigma_n$ naturally extends to $K[y_1,\dots,y_n,$
  $x_1,\dots,x_{\binom{n}{2}},$ $u_1,\dots,u_m]$.
\end{definition}

\begin{definition}
  A substitution family $\sigma = (\sigma_n)$ is \emph{constant-time
    computable} if given $n$ and a variable $z$ in the domain of
  $\sigma_n$, we can compute $\sigma_n(z)$ in constant-time on a
  Word-RAM. (Note that an encoding of any $z$ fits into on cell of
  memory.)
\end{definition}


\begin{definition} \label{def:reduc-full}
  Let $f = (f_n)$ and $g = (g_n)$ be graph pattern polynomial
  families. Then $f$ is reducible to $g$, denoted $f\redto g$, via a
  constant time computable substitution family $\sigma = (\sigma_n)$
  if for all $n$ there is an $m = O(n)$ and $q = O(1)$ such that
  \begin{enumerate}
    \item $\sigma_m(\ml{g_m})$ is a graph pattern polynomial and
    \item $\ml{\sigma_m(g_m)} = v_{[q]}\ml{f_n}$. (Recall that
      $v_{[q]} = v_1 \cdots v_q$.)
  \end{enumerate}
For any prime $p$, we say that $f\redto g \pmod p$ if there exists an
$f' = f \pmod p$ such that $f'\redto g$.
\end{definition}

Property~1 of Definition~\ref{def:reduc-full} and Properties~1 and 3 of
Definition~\ref{def:subfam-full} imply that $\sigma_m(g_m)$ is a graph
pattern polynomial because Properties~1 and 3 of
Definition~\ref{def:subfam-full} ensure that non-multilinear terms remain
so after the substitution. It is easy to see that $\redto$ is
reflexive via the identity substitution. We can also assume
w.l.o.g.\ that the variables $v_1,\dotsc,v_q$ are fresh variables
introduced by the substitution family $\sigma$.

What is the difference between $\sigma_m(\ml{g_m})$ and
$\ml{\sigma_m(g_m)}$ in the Definition~\ref{def:reduc-full}?  Every
monomial in $\ml{\sigma_m(g_m)}$ also appears in $\sigma_m(\ml{g_m})$,
however, the latter may contain further monomials that are not
multilinear.

\begin{proposition}
  $\redto$ is transitive.
\end{proposition}
\begin{proof}
  Let $f\redto g$ via $\sigma$ and $g\redto h$ via $\tau$.  Assume
  that $f_n$ is written as a substitution instance of $g_{m(n)}$ by
  $\sigma$ and $g_m$ is written as a substitution instance of
  $h_{r(m)}$ by $\tau$ for some linearly bounded functions $m$ and
  $r$.  Let $\sigma_{m(n)}(g_{m(n)})$ and
  $\tau_{r(m(n))}(h_{r(m(n)))})$ have $u_1,\dotsc,u_p$ and
  $v_1,\dotsc,v_q$, respectively, as other variables that are
  multiplied with the multilinear terms. We can assume w.l.o.g.\ that
  these two sets of other variables are disjoint.

  Define $\sigma'$ as $\sigma$ extended to $v_i$ by $\sigma'_n(v_i) =
  v_i$ for all $i$ and $n \in \mathbb{N}$.  We claim that $f\redto h$
  via the family $(\sigma'_{m(n)} \circ \tau_{r(m(n))})$.  We need to
  verify the two properties of Definition \ref{def:reduc-full}.

  \emph{Property~1}:
  $\sigma'_{m(n)}(\tau_{r(m(n))}(\ml{h_{r(m(n))}})) =
  \sigma'_{m(n)}(v_{[q]}\ml{g_{m(n)}} + h')$ where $h'$ is a graph
  pattern polynomial containing only non-multilinear terms. Now, we
  have $h'' = \sigma'_{m(n)}(v_{[q]}\ml{g_{m(n)}})$
  $= v_{[q]}\sigma_{m(n)}(\ml{g_{m(n)}})$ because $\ml{g_{m(n)}}$
  cannot contain $v_i$ and $\sigma'_{m(n)}(v_i) = v_i$ for $i\in
  [q]$. This implies that $h''$ is a graph pattern polynomial because
  $\sigma_{m(n)}(\ml{g_{m(n)}})$ is a graph pattern polynomial. Also,
  $\sigma'_{m(n)}(h')$ is a graph pattern polynomial containing only
  non-multilinear terms by Properties~1 and 3 of
  Definition~\ref{def:subfam-full} proving that
  $(\sigma'_{m(n)} \circ \tau_{r(m(n))})(\ml{h_{r(m(n))}})$ is a graph
  pattern polynomial.

  \emph{Property~2} is proved as follows:
  
  \begin{align*}
    \ml{(\sigma'_{m(n)} \circ \tau_{r(m(n))})(h_{r(m(n))})}
    &= \ml{\sigma'_{m(n)}(\tau_{r(m(n))}(h_{r(m(n)))}))}\\
    &= \ml{\sigma'_{m(n)}(v_{[q]}\ml{g_{m(n)}} + h')}\\
    &= \ml{v_{[q]}\sigma_{m(n)}(\ml{g_{m(n)}})}\\
    &= v_{[q]}\ml{\sigma_{m(n)}(\ml{g_{m(n)}})}\\
    &= v_{[q]}u_{[p]}\ml{f_n }
  \end{align*}

  Note that the term $h'$ vanishes, since $\sigma_{m(n)}$ does not
  introduce new multilinear monomials and also $\ml{.}$ is a linear
  operator. The same happens in the second-last line, we did not write
  the additional term in the equation, since it vanishes anyway.

  We also have $r(m(n)) = O(n)$ and $p + q = O(1)$. It is easy to
  verify that $(\sigma'_{m(n)} \circ \tau_{r(m(n))})$ is a
  constant-time computable substitution family.
\end{proof}

Efficient algorithms are known for detecting multilinear terms of
\emph{small} degree with non-zero coefficient modulo primes. We state
two such theorems that we use in this paper.

\begin{theorem}
  \label{thm:test:1-full}
  Let $k$ be any constant and let $p$ be any prime. Given an
  arithmetic circuit of size $s$, there is a randomized, one-sided
  error $O(s)$-time algorithm to detect whether the polynomial
  computed by the circuit has a multilinear term of degree atmost $k$
  with non-zero modulo $p$ coefficient.
\end{theorem}

\begin{theorem}
  \label{thm:multiparity-full}
  Let $k$ be any constant. Given an arithmetic circuit of size $s$
  computing a polynomial of degree $k$ on $n$ variables, there is a
  deterministic $O(s+n^{\lceil k/2 \rceil})$-time algorithm to compute
  the parity of the sum of coefficient of multilinear terms.
\end{theorem}

An important algorithmic consequence of reducibility is stated in
Proposition~\ref{prop:red-alg-full}.

\begin{proposition}
  Let $p$ be any prime. Let $f$ and $g$ be graph pattern polynomial
  families. Let $s(n)$ be a polynomially-bounded function. If
  $f\redto g$ and $g$ has size uniform $s(n)$-size arithmetic
  circuits, then we can test whether $f_n(G)$ has a multilinear term
  with non-zero coefficient modulo $p$ in $O(s(n))$ (randomized
  one-sided error) time for any $n$-vertex graph $G$.
  \label{prop:red-alg-full}
\end{proposition}
\begin{proof}
  Assume that $f_n$ is reducible to $g_m$ where $m = O(n)$. Since
  $s(n)$ is polynomially bounded, we have $\size(g_m) =
  O(s(n))$. Apply the substitution $\sigma_m$ to $g_m$ to obtain
  $g'$. Let $u_1,\dots,u_r$ be the other variables of $g'$. We claim
  that testing whether the polynomial $g'(G)$ has a multilinear term
  is equivalent to testing whether $f_n(G)$ has a multilinear term. We
  have $u_{[r]}\ml{f_n} = \ml{g'}$. Since both $f_n$ and $g'$ are
  graph pattern polynomials, we have
  $u_{[r]}\ml{f_n(G)} = u_{[r]}\ml{f_n}(G) = \ml{g'}(G) =
  \ml{g'(G)}$. Therefore, testing whether the polynomial $f_n(G)$ has
  a multilinear term of degree at most $k$, where $k$ is some
  constant, reduces to testing whether $g'(G)$ has a multilinear term
  of degree $k + r = O(1)$. Since $g'$ has $O(s(n))$ size circuits,
  this can be done in $O(s(n))$ (randomized one-sided error) time.
\end{proof}

On the other hand, if we only have $f\redto g \pmod p$ for some
specific prime $p$, then it is only possible to test for multilinear
terms in $f$ that have non-zero coefficients modulo $p$ for that prime
$p$.

\begin{corollary}
  Let $f\redto g \pmod p$ and $g$ has $s(n)$ size circuits where
  $s(n)$ is polynomially bounded. Then we can test whether $f_n(G)$ has
  a multilinear term with non-zero coefficient modulo $p$ in $O(s(n))$
  time for any $n$-vertex graph $G$.
\end{corollary}

More relaxed notions of reduction allowing an increase of
$\polylog(n)$ factors in size or allowing multilinear terms to be
multiplied by arbitrary sets of other variables could also be useful
to obtain better algorithms. We do not pursue this because we could
not find any reductions that make use of this freedom.

The following result allows efficient construction of $\ghom{H}$ when
$H$ has small treewidth.

\begin{theorem}{(Implicit in \cite{DiazST02}, Also used in \cite{FominLRSR12} and \cite{DMMRS14})}
  $\ghom{H}$ can be computed by $O\bigl(n^{\tw{H}+1}\bigr)$ size
  arithmetic circuits for all graphs $H$.
  \label{thm:diaz-full}
\end{theorem}


\section{Pattern graphs easier than cliques}
\label{sec:easypattern-full}
In this section, we describe a family $H_{3k}$ of pattern graphs such that the
induced subgraph isomorphism problem for $H_{3k}$ has an $O(n^{\omega(k, k-1,
k)})$ time algorithm when $k = 2^\ell, \ell \geq 1$. Note that for the currently
known best algorithms for fast matrix multiplication, we have $\omega(k, k-1, k)
< k\omega$. Therefore, these pattern graphs are strictly easier to detect than
cliques.

The pattern graph $H_{3k}$ is defined on $3k$ vertices and we consider the
canonical labelling of $H_{3k}$ where there is a $(3k-1)$-clique on vertices
$\{1,\dotsc,3k-1\}$ and the vertex $3k$ is adjacent to the vertices
$\{1,\dotsc,2k-1\}$.

\begin{lemma}
    \label{lem:IH3kNH3kmod2-full}
  $\ind{H_{3k}} = \sub{H_{3k}} \pmod 2$ when $k = 2^\ell, \ell \geq 1$
\end{lemma}
\begin{proof}
  We show that the number of times $H_{3k}$ is contained in any of its proper
  supergraphs is even if $k$ is a power of $2$.  The graph $K_{3k}$ contains
  $3k{\binom{3k-1}{2k-1}}$ copies of $H_{3k}$. This number is even when $k$ is
  even. The graph $K_{3k} - e$ contains $2{\binom{3k-2}{2k-1}}$ copies of
  $H_{3k}$. This number is always even. The remaining proper supergraphs of
  $H_{3k}$ are the graphs $K_{3k-1} + (2k+i)e$, i.e., a $(3k-1)$-clique with
  $2k+i$ edges to a single vertex, for $0\leq i < k-2$. There are $m_i =
  {\binom{2k+i}{2k-1}}$ copies of the graph $H_{3k}$ in these supergraphs. We
  observe that the numbers $m_i$ are even when $k = 2^\ell, \ell \geq 1$ by
  Lucas' theorem. Lucas' theorem states that ${\binom{p}{q}}$ is even if and
  only if in the binary representation of $p$ and $q$, there exists some bit
  position $i$ such that $q_i = 1$ and $p_i = 0$. To see why $m_i$ is even,
  observe that in the binary representation of $2k-1$, all bits $0$ through
  $\ell$ are $1$ and in the binary representation of $2k+i, 0\leq i < k-2$, at
  least one of those bits is $0$.
\end{proof}

\begin{lemma}
  \label{lem:NH3kHomH3k-full}
  $\sub{H_{3k}} \redto \ghom{H_{3k}}$
\end{lemma}
\begin{proof}
  We start with $\ghom{H_{3k}}$ over the vertex set $[n]\times [3k]$ and apply
  the following substitution.

  \setcounter{equation}{0}
  \begin{align}
    \sigma(z_{a, (v, a)}) &= z_a\\
    \sigma(z_{a, (v, b)}) &= z_a^2, a\neq b\\
    \sigma(y_{(v, a)}) &= y_v\\
    \sigma(x_{(u, a), (v, b)}) &= 0, \text{if $a, b\in\{1,\dotsc,2k-1\}$ and $a < b$ and $u > v$}\\
    \sigma(x_{(u, a), (v, b)}) &= 0, \text{if $a, b\in\{2k,\dotsc,3k-1\}$ and $a < b$ and $u > v$}\\
    \sigma(x_{(u, a), (v, b)}) &= x_{\{u, v\}}, \text{otherwise}
  \end{align}

  Rule 3 ensures that in any surviving monomial, all vertices have distinct
  labels. Rule 4 ensures that the vertices coloured $1,\dotsc, 2k-1$ are in
  increasing order and Rule 5 ensures that the vertices coloured
  $2k,\dotsc,3k-1$ are in increasing order.

  Consider an $H_{3k}$ labelled using $[n]$ where the vertices in the
  $(3k-1)$-clique are labelled $v_1,\dotsc, v_{3k-1}$ and the
  remaining vertex is labelled $v_{3k}$ which is connected to
  $v_1 < \dotso < v_{2k-1}$. Also, $v_{2k} < \dotso < v_{3k-1}$. We
  claim that the monomial corresponding to this labelled $H_{3k}$ (say
  $m$) is uniquely generated by the monomial
  $m' = \prod_{1\leq i\leq 3k} z_{i, (v_i, i)} w$ in
  $\ghom{H_{3k}}$. Note that the vertices and edges in the image of
  the homomorphism is determined by the map $i\mapsto (v_i, i)$. The
  monomial $w$ is simply the product of these vertex and edge
  variables. It is easy to see that this monomial yields the required
  monomial under the above substitution.  The uniqueness is proved as
  follows: observe that in any monomial $m''$ in $\ghom{H_{3k}}$ that
  generates $m$, the vertex coloured $3k$ must be $v_{3k}$.  This
  implies that the vertices coloured $1,\dotsc,2k-1$ must be the set
  $\{v_1,\dotsc,v_{2k-1}\}$. Rule~4 ensures that vertex coloured $i$
  must be $v_i$ for $1\leq i \leq 2k-1$. Similarly, the vertices
  coloured $2k,\dotsc,3k-1$ must be the set
  $\{v_{2k},\dotsc,v_{3k-1}\}$ and Rule~5 ensures that vertex coloured
  $i$ must be $v_i$ for $2k\leq i\leq 3k-1$ as well. But then the
  monomials $m'$ and $m''$ are the same.
\end{proof}

\begin{lemma}
\label{lem:HomH3kcircuits-full}
  $\ghom{H_{3k}}$ can be computed by arithmetic circuits of size
  $O(n^{\omega(k, k-1, k)})$ for $k > 1$.
\end{lemma}
\begin{proof}
  Consider $H_{3k}$ labelled as before.  We define the sets
  $S_{1,k,2k,3k-1} = \{1,\dotsc,k,2k\dotsc,3k-1\}$, $S_{k+1, 3k-1} =
  \{k+1,\dotsc,3k-1\}$, $S_{k+1, 2k-1} = \{k+1,\dotsc,2k-1\}$, and
  $S_{1, 2k-1} = \{1,\dotsc,2k-1\}$. We also define the tuples $V_{1,
    k} = (v_1, \dotsc, v_k)$, $V_{2k, 3k-1} = (v_{2k}, \dotsc,
  v_{3k-1})$, and $V_{k+1, 2k-1} = (v_{k+1}, \dotsc, v_{2k-1})$ for
  any set $v_i$ of $3k-1$ distinct vertex labels. The algorithm also
  uses the matrices defined below. The dimensions of each matrix are
  specified as the superscript. All other entries of the matrix are
  $0$.

  \begin{align*}
    A^{n^k\times n^k}_{V_{1, k}, V_{2k, 3k-1}} &= \smashoperator{\prod\limits_{i\in S_{1,k,2k,3k-1}}} z_{i, v_i} y_{v_i} \prod\limits_{\substack{i, j\in S_{1,k,2k,3k-1}\\ i\neq j}} x_{\{v_i, v_j\}}
    ,\text{$v_i$ distinct for $1\leq i \leq 3k-1$}\\
    B^{n^k\times n^{k-1}}_{V_{2k, 3k-1}, V_{k+1, 2k-1}} &= \smashoperator{\prod\limits_{i\in S_{k+1,2k-1}}} z_{i, v_i} y_{v_i} \prod\limits_{\substack{i\in S_{k+1,3k-1}\\ j\in S_{k+1,2k-1}\\ i\neq j}} x_{\{v_i, v_j\}}
    ,\text{$v_i$ distinct for $k+1\leq i\leq 3k-1$}\\
    C^{n^{k-1}\times n^k}_{V_{k+1,2k-1}, V_{1,k}} &= x_{\{{(v_i, i)}_{i\in S_{1,2k-1}}\}} \smashoperator{\prod\limits_{\substack{i\in S_{k+1,2k-1}\\ j\in [k]\\ i\neq j}}} x_{\{v_i, v_j\}}
    ,\text{$v_i$s are distinct for $1\leq i\leq 2k-1$}\\
    D^{n^k\times n}_{V_{1,k}, v_{3k}} &= z_{3k, v_{3k}} y_{v_{3k}} \prod\limits_{i\in [k]} x_{\{v_i, v_{3k}\}}
    , \text{$v_i$ distinct for $i\in\{1,\dotsc,k,3k\}$}\\
    E^{n\times n^{k-1}}_{v_{3k}, V_{k+1,2k-1}} &= \smashoperator{\prod\limits_{i\in S_{k+1,2k-1}}} x_{\{v_i, v_{3k}\}}
    , \text{$v_i$ distinct for $i\in\{k+1,\dotsc,2k-1,3k\}$}
  \end{align*}

  Compute the matrix products $ABC$ and $DE$. Replace the $n^{2k-1}$
  variables $x_{\{{(v_i, i)}_{i\in S_3}\}}$ with
  ${(DE)}_{V_{1,k}, V_{k+1,2k-1}}$. The required
  polynomial is then just

  \begin{align*}
    \ghom{H_{3k}} &= \sum_{(v_1, \dotsc, v_k)} {(ABC)}_{(v_1, \dotsc,
      v_k),(v_1, \dotsc, v_k)}
  \end{align*}

  Consider a homomorphism of $H_{3k}$ defined as $\phi: i\mapsto
  u_i$. The monomial corresponding to this homomorphism is uniquely
  generated as follows.  Let $U_{*}$ be defined similarly to the
  tuples $V_{*}$. Set $v_i = u_i$ for $i\in [k]$ in the summation and
  consider the monomial generated by the product $A_{U_{1,k},
    U_{2k,3k-1}}B_{U_{2k,3k-1}, U_{k+1, 2k-1}}C_{U_{k+1, 2k-1}, U_{1,
      k}}$ after replacing the variable $x_{\{{(u_i, i)}_{i\in
      S_3}\}}$ by ${(DE)}_{U_{1, k}, U_{k+1,2k-1}}$ taking the
  monomial $D_{U_{1,k},u_{3k}}E_{u_{3k}, U_{k+1, 2k-1}}$ from that
  entry. It is easy to verify that this generates the required
  monomial. For uniqueness, observe that this is the only way to
  generate the required product of the homomorphism variables.
    
  Computing $ABC$ can be done using $O(n^{\omega(k, k-1, k)})$ size circuits.
  Computing $DE$ can be done using $O(n^{\omega(k, 1, k-1)})$ size circuits. The
  top level sum contributes $O(n^k)$ gates. This proves the lemma.
\end{proof}

We conclude this section by stating our main theorem.

\begin{theorem}\label{thm:H3k-full}
  The induced subgraph isomorphism problem for $H_{3k}$ has an $O(n^{\omega(k,
  k-1, k)})$ time algorithm when $k = 2^\ell, \ell \geq 1$.
\end{theorem}

\section{Algorithms for induced paths and cycles}

In this section, we will prove that the time complexity of the induced
subgraph isomorphism problems for paths and cycles are upper bounded
by the circuit complexities of the homomorphism polynomials for
$\overline{P_k}$ and $K_k - P_{k-1}$ respectively. Using this we
derive efficient algorithms for induced subgraph isomorphism problem
for $P_k$ for $k\in \{5,6,7,8,9\}$ and $C_k$ for $k\in \{5,7,9\}$. We
also obtain efficient combinatorial algorithms for the induced
subgraph isomorphism problem for $P_k$ for all $k$ and $C_k$ when $k$
is odd.

The proof has two main steps: First, we show that the induced subgraph
isomorphism polynomials for these patterns are reducible to the
aforementioned homomorphism polynomials
(Lemmas~\ref{lem:path-ind-sub-full}, \ref{lem:path-hom-full},
\ref{lem:cycle-ind-sub-full}, \ref{lem:cycle-hom-full}). Then, we
prove that these homomorphism polynomials can be computed efficiently
(Theorems~\ref{thm:path-algos-full} and \ref{thm:cycle-algos-full}).

\begin{lemma}
  $\ind{\overline{P_k}} = \sub{\overline{P_k}} \pmod 2$ for $k\geq 4$.
  \label{lem:path-ind-sub-full}
\end{lemma}
\begin{proof}
  We will prove that for any proper super-graph $H$ of
  $\overline{P_k}$, the number $\nsub{\overline{P_k}}{H}$ is even.
  Observe that this number is the same as the number of ways to extend
  a proper labelled subgraph of $P_k$ to some labelled $P_k$. Let $H$
  be an arbitrary proper subgraph of $P_k$. Let $2\leq \ell \leq k$ be
  the number of connected components in $H$ out of which
  $0\leq s \leq \ell$ of them consists only of a single vertex. Then
  the number of ways to extend $H$ to a $P_k$ is
  $\ell!2^{\ell-s}/2$. We can extend $H$ to a $P_k$ by ordering the
  connected components from left to right and then connecting the
  endpoints from left to right. There are $\ell!$ ways to order $\ell$
  components and $2$ ways to place all components with more than one
  vertex. Out of these, a configuration and its reverse will lead to
  the same labelled $P_k$. Since $\ell \geq 2$, this number is even if
  $\ell > s$. Otherwise, this number is $k!/2$ because $\ell = s$
  implies that there are $k$ components. This is even when $k\geq
  4$. We conclude that
  $\ind{\overline{P_k}} = \sub{\overline{P_k}} \pmod 2$.
\end{proof}

\begin{lemma}
  $\sub{\overline{P_k}} \redto \ghom{\overline{P_k}}$
  \label{lem:path-hom-full}
\end{lemma}
\begin{proof}
  Let $f = \sub{\overline{P_k}}$ and $g = \ghom{\overline{P_k}}$. We
  fix the labelling of $\overline{P_k}$ where the vertices of the
  complementary $P_k$ are labelled $1,2,\dotsc,k$ with $1$ and $k$ as
  the endpoints and for every other vertex $i$, the neighbours are
  $i-1$ and $i+1$. Start with $g$ over the vertex set $[n]\times [k]$
  and use the following substitution.

  \setcounter{equation}{0}
  \begin{align}
    \sigma(z_{a,(v,a)}) &= z_a\\
    \sigma(z_{a,(v,b)}) &= z_a^2 \text{, if $a\neq b$}\\
    \sigma(y_{(v, a)}) &= y_v\\
    \sigma(x_{\{(u, p), (v, q)\}}) &= 0 \text{, if $\{p,q\}\not\in E(\overline{P_k})$ or if $p=1$ and $q=k$ and $u>v$}\\
    \sigma(x_{\{(u, p), (v, q)\}}) &= x_{\{u, v\}} \text{, otherwise}
  \end{align}

  The resulting polynomial $g'$ satisfies
  $\ml{g'} = z_1\ldots z_k\ml{f_n}$ as required. The reduction works
  because there is exactly one non-trivial automorphism for
  $\overline{P_k}$ and that automorphism maps $1$ to $k$. The monomial
  corresponding to one of these automorphisms become $0$ because of
  $u>v$ where $u$ has colour 1 and $v$ has colour $k$.
\end{proof}

\begin{theorem}
  If $\ghom{\overline{P_k}}$ can be computed by circuits of size
  $n^{f(k)}$, then there is an $O(n^{f(k)})$ time algorithm for the
  induced subgraph isomorphism problem for $P_k$ on $n$-vertex graphs.
  \label{thm:path-main-theorem-full}
\end{theorem}

\begin{theorem}
  \label{thm:AlgosForPaths-full}
  The following algorithms exist

  \begin{enumerate}
  \item An $O(n^\omega)$-time algorithm for induced subgraph
    isomorphism problem for $P_5$ in $n$-vertex graphs.
  \item An $O(n^{\omega(2,1,1)})$-time algorithm for induced subgraph
    isomorphism problem for $P_6$ in $n$-vertex graphs.
  \item An $O(n^{k-2})$-time combinatorial algorithm for induced subgraph
    isomorphism problem for $P_k$ in $n$-vertex graphs.
  \item An $O(n^{k-2})$-time deterministic combinatorial algorithm for
    computing the parity of the number of induced subgraphs isomorphic
    to $P_k$ in $n$-vertex graphs.
  \end{enumerate}
  \label{thm:path-algos-full}
\end{theorem}
\begin{proof}
  \begin{enumerate}
  \item We describe how to compute $\ghom{\overline{P_5}}$ using
    arithmetic circuits of size $O(n^\omega)$. We start by defining
    the following matrices.
    
    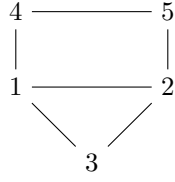
\begin{figure}
      \begin{center}
      \begin{tikzpicture}
        \node (1) at (0, 0) {1};
        \node (2) at (2, 0) {2};
        \node (5) at (2, 1) {5};
        \node (3) at (1, -1) {3};
        \node (4) at (0, 1) {4};
        
        \draw (1) -- (2);
        \draw (1) -- (4);
        \draw (1) -- (3);

        \draw (2) -- (3);
        \draw (2) -- (5);
        \draw (4) -- (5);
      \end{tikzpicture}
      \end{center}
      \caption{A labelled $\overline{P_5}$}
      \label{fig:p5bar-full}
    \end{figure}

    \begin{align*}
      A^{n\times n}_{i, j} &= x_{\{i, j\}}, i\neq j\\
      B^{n\times n}_{i, i} &= y_i z_{3, i}\\
      C^{n\times n}_{i, i} &= y_i z_{4, i}\\
      D^{n\times n}_{i, i} &= y_i z_{5, i}
    \end{align*}

    Consider the labelled $\overline{P_5}$ in Figure~\ref{fig:p5bar-full}.
    Then we can write
    \[
      \ghom{\overline{P_5}} = \sum_{i, j\in[n], i\neq j} z_{1, i} z_{2,
      j} x_{\{i, j\}} y_i y_j {(ABA)}_{i, j} {(ACADA)}_{i,
      j}
  \]
  Clearly, this can be implemented using $O(n^\omega)$ size
  circuits. We will now prove that this circuit correctly computes the
  polynomial $\ghom{\overline{P_5}}$. Consider a homomorphism
  $\phi: j\mapsto i_j$. Consider the monomial generated by
  $i = i_1, j = i_2$ in the outer sum, the monomial
  $A_{i_1, i_3}B_{i_3, i_3}A_{i_3, i_2}$ in the product
  ${(ABA)}_{i_1, i_2}$, and the monomial
  $A_{i_1, i_4}C_{i_4, i_4}$ $A_{i_4, i_5}D_{i_5,i_5}A_{i_5,i_2}$ in the
  product ${(ACADA)}_{i_1, i_2}$. This monomial corresponds to the
  homomorphism $\phi$ and one can observe that this is the only way to
  generate this monomial. On the other hand, any monomial in the
  computed polynomial is generated as described above and therefore
  corresponds to a homomorphism.
    
  \item We show how to compute $\ghom{\overline{P_6}}\,$ using
    arithmetic circuits of size $O(n^{\omega(2, 1, 1)})$. We define
    the following matrices.
    \begin{align*}
      A^{n\times n^2}_{i, (j, k)} &=
                                    z_{2, i}z_{1, j}z_{6, k} y_i y_j y_k
                                    x_{\{(2, i), ((1, j), (6, k))\}} x_{\{j, k\}} x_{\{k, i\}},
                                    j\neq k, i\neq k\\
      B^{n^2\times n}_{(j, k), \ell} &=
                                       z_{5,\ell} y_\ell x_{\{((1, j), (6, k)), (5, \ell)\}}
                                       x_{\{j, \ell\}}, j\neq k, j\neq \ell\\
      C^{n\times n}_{\ell, i} &= x_{\{\ell, i\}}, \ell\neq i\\
      D^{n^2\times n}_{(j, k), p} &= y_p z_{3, p} x_{\{((1, j), (6, k)), (3, p)\}},
                                    j\neq k, j\neq p, k\neq p\\
      E^{n\times n}_{p,\ell} &= x_{\{p, \ell\}}, p\neq \ell\\
      F^{n^2\times n}_{(j, k), q} &= y_q z_{4, q} x_{\{((1, j), (6, k)), (4, q)\}},
                                    j\neq k, j\neq q, k\neq q\\
      G^{n\times n}_{q,i} &= x_{\{q, i\}}, q\neq i
    \end{align*}

    Compute the matrix products $ABC$, $DE$, and $FG$.  The output of
    the circuit is $\sum_i {(ABC)}_{i, i}$ after substituting for the
    variables as follows. Replace each $x_{\{((1, j), (6,
      k)), (5, \ell)\}}$ with $DE_{(j, k), \ell}$ and each
     $x_{\{((1, j), (6, k)), (2, i)\}}$ with $FG_{(j, k),
      i}$. Replace each $x_{\{((1, j), (6, k)), (3, p)\}}$
    with $x_{\{j, p\}} x_{\{k, p\}}$ and each
    $x_{\{((1, j), (6, k)), (4, q)\}}$ with $x_{\{j, q\}} x_{\{k,
      q\}}$.

    \begin{figure}
      \begin{center}
      \begin{tikzpicture}
        \node (1) at (0, 0) {1};
        \node (6) at (0, 2) {6};
        \node (2) at (1, -0.5) {2};
        \node (5) at (1, 2.5) {5};
        \node (3) at (2, 0.5) {3};
        \node (4) at (2, 1.5) {4};
        
        \draw (1) -- (6);
        \draw (1) -- (5);
        \draw (1) -- (4);
        \draw (1) -- (3);

        \draw (2) -- (4);
        \draw (2) -- (5);
        \draw (2) -- (6);
        
        \draw (3) -- (5);
        \draw (3) -- (6);
        
        \draw (4) -- (6);
      \end{tikzpicture}
      \end{center}
      \caption{A labelled $\overline{P_6}$}
      \label{fig:p6bar-full}
    \end{figure}
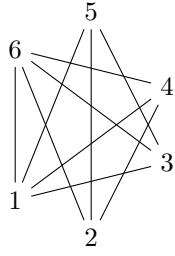

    Consider the labelling of $\overline{P_6}$ in
    Figure~\ref{fig:p6bar-full}. After substituting for all variables as
    mentioned above, the monomials of ${(ABC)}_{i, i}$ correspond to
    homomorphisms from this labelled $\overline{P_6}$ to $K_n$ that
    maps vertex $2$ to $i$. Therefore, the circuit correctly computes
    $\ghom{\overline{P_6}}$.

  \item We observe that $\tw{\overline{P_k}} = k-3$ and therefore
    using Theorem~\ref{thm:diaz-full}, we can compute
    $\ghom{\overline{P_k}}$ using $O(n^{k-2})$ size circuits.
    
  \item Consider the substitution in the proof of
    Lemma~\ref{lem:path-hom-full} and replace rules (1) and (2) by the
    following rules.

    \begin{align*}
      \sigma(z_{a, (v, a)}) &= 1 \tag{1'}\\
      \sigma(z_{a, (v, b)}) &= 0 \tag{2'}
    \end{align*}

    The multilinear part of the resulting polynomial $f$ is the same
    as $\sub{\overline{P_k}}$ and hence has degree-$k$. Therefore, we
    only have to compute the parity of the sum of coefficients of the
    multilinear terms of $f(G)$. By Theorem~\ref{thm:multiparity-full},
    this can be done in $O(n^{k-2})$ time.
  \end{enumerate}
\end{proof}

We remark that by computing homomorphism polynomials for
$\overline{P_{k}}$ for $k = 7, 8, 9$ using small-size circuits, we can
obtain the following algorithms for the induced subgraph isomorphism
problem for paths: An $O(n^{2\omega})$ time algorithm for $P_7$, an
$O(n^{\omega(3, 2, 2)})$ time algorithm for $P_8$, and an
$O(n^{\omega(3,3,2)})$ time algorithm for $P_9$. All these algorithms
are faster than the corresponding algorithms for $k$-cliques.

\begin{lemma}
  $\ind{\overline{C_k}} = \sub{\overline{C_k}} + \sub{\overline{P_k}}
  + \sub{K_k - P_{k-1}} \pmod 2$ for $k\geq 5$.
  \label{lem:cycle-ind-sub-full}
\end{lemma}
\begin{proof}
  We claim that the only proper supergraphs of $\overline{C_k}$
  containing it an odd number of times are $\overline{P_k}$ and
  $K_k - P_{k-1}$. There is exactly one way to extend a $P_k$ or a
  $P_{k-1} + v$ to a $C_k$. Let $H$ be a proper subgraph of $C_k$
  other than these two graphs. Assume that $H$ has
  $2 \leq \ell \leq k$ connected components out of which
  $0\leq s\leq \ell$ are single vertices. Then there are
  $m = \ell! 2^{\ell-s}/2\ell$ ways to extend $H$ to $C_k$. If
  $\ell > s$, then $m$ is even because $(\ell-1)!$ is even when
  $\ell\geq 3$ and when $\ell = 2$ the number $s$ is $0$ and $m =
  2$. If $\ell = 2$ and $s = 1$, then $H = P_{k-1} + v$. If
  $\ell = s$, then $m = \ell!/2\ell = (\ell-1)!/2$. But $\ell = s$
  implies that $\ell = k$ and therefore $m = (k-1)!/2$ which is even
  when $k\geq 5$.
\end{proof}

\begin{lemma}
  \begin{enumerate}
  \item $\sub{\overline{C_k}} \redto \ghom{K_k - P_{k-1}} \pmod 2$
    for odd $k\geq 5$.
  \item $\sub{\overline{P_k}} \redto \ghom{K_k - P_{k-1}}$ for
    $k\geq 5$.
  \item $\sub{K_k - P_{k-1}} \redto \ghom{K_k - P_{k-1}}$ for
    $k\geq 5$.
  \item $\ind{\overline{C_k}} \redto \ghom{K_k - P_{k-1}} \pmod 2$ for
    odd $k\geq 5$.
  \end{enumerate}
  \label{lem:cycle-hom-full}
\end{lemma}
\begin{proof}
  We start with $\ghom{K_k - P_{k-1}}$ over the vertex set
  $[n]\times [k]$ in all cases and apply the following substitutions.
  \begin{enumerate}
  \item Fix the labelling of $\overline{C_k}$ where the complementary
    $C_k$ is labelled $1,\dotsc,k$ such that the vertex $1$ has
    neighbours $2$ and $k$ and $k$ has neighbours $1$ and $k-1$ and
    every other vertex $i$ has $i+1$ and $i-1$ as its neighours. The
    crucial observation is that $\overline{C_k}$ has $2k$
    automorphisms and if we only select automorphisms where the label
    of the vertex coloured $1$ is strictly less than the label of the
    vertex coloured $3$, then we select exactly $k$
    automorphisms. This allows us to compute a polynomial family $h$
    such that $k.\sub{\overline{C_k}} \redto h$ and
    $k.\sub{\overline{C_k}} = \sub{\overline{C_k}} \pmod 2$.

    \setcounter{equation}{0}
    \begin{align}
      \sigma_1(z_{a,(v,a)}) &= z_a\\
      \sigma_1(z_{a,(v,b)}) &= z_a^2 \text{, if $a\neq b$}\\
      \sigma_1(y_{(v, a)}) &= y_v\\
      \sigma_1(x_{\{(u, p), (v, q)\}}) &= 0 \text{, if $p=1$ and $q=3$ and $u>v$}\\
      \sigma_1(x_{\{(u, p), (v, q)\}}) &= 1 \text{, if $p=1$ and $q=2$ or $p=1$ and $q=k$}\\
      \sigma_1(x_{\{(u, p), (v, q)\}}) &= x_{\{u, v\}} \text{, otherwise}
    \end{align}
  \item
    
    Fix the labelling of $\overline{P_k}$ where the complementary
    $P_k$ is
    \tikz[baseline=-\the\dimexpr\fontdimen22\textfont2\relax,node distance=0.5cm]{\node
      (1) {1};\node[right of=1] (2) {2};\node[right of=2] (dots) {$\dotsm$};\node[right of=dots] (k) {$k$};
      \draw (1) -- (2) -- (dots) -- (k);}.
    
    \setcounter{equation}{0}
    \begin{align}
      \sigma_2(z_{a,(v,a)}) &= z_a\\
      \sigma_2(z_{a,(v,b)}) &= z_a^2 \text{, if $a\neq b$}\\
      \sigma_2(y_{(v, a)}) &= y_v\\
      \sigma_2(x_{\{(u, p), (v, q)\}}) &= 0 \text{, if $p=1$ and $q=k$ and $u>v$}\\
      \sigma_2(x_{\{(u, p), (v, q)\}}) &= 1 \text{, if $p=1$ and $q=2$}\\
      \sigma_2(x_{\{(u, p), (v, q)\}}) &= x_{\{u, v\}} \text{, otherwise}
    \end{align}

  \item Fix the labelling of $K_k - P_{k-1}$ where the complementary
    $P_{k-1} + v$ is
    \tikz[baseline=-\the\dimexpr\fontdimen22\textfont2\relax,node
    distance=0.5cm]{\node (1) {1};\node[right of=1] (2) {2};\node[right
      of=2] (3) {3};\node[right of=3] (dots) {$\dotsm$};\node[right
      of=dots] (k) {$k$}; \draw (2) -- (3) -- (dots) -- (k);}.

    \setcounter{equation}{0}
    \begin{align}
      \sigma_3(z_{a,(v,a)}) &= z_a\\
      \sigma_3(z_{a,(v,b)}) &= z_a^2 \text{, if $a\neq b$}\\
      \sigma_3(y_{(v, a)}) &= y_v\\
      \sigma_3(x_{\{(u, p), (v, q)\}}) &= 0 \text{, if $p=2$ and $q=k$ and $u>v$}\\
      \sigma_3(x_{\{(u, p), (v, q)\}}) &= x_{\{u, v\}} \text{, otherwise}
    \end{align}

  \item We prove that
    $k\sub{\overline{C_k}} + \sub{\overline{P_k}} + \sub{K_k -
      P_{k-1}} \redto \ghom{K_k - P_{k-1}}$. Start with
    $\ghom{K_k - P_{k-1}}$ over the vertex set
    $[n]\times [k] \times [3]$ and apply the following substitution.
    \setcounter{equation}{0}
    \begin{align}
      \sigma(z_{a,(v,b,i)}) &= \sigma_i(z_{a, (v, b)}))\\
      \sigma(y_{(v, a, i)}) &= \sigma_i(y_{(v, a)})\\
      \sigma(x_{\{(u, p, i), (v, q, j)\}}) &= 0 \text{, if $i\neq j$}\\
      \sigma(x_{\{(u, p, i), (v, q, i)\}}) &= \sigma_i(x_{\{(u, p), (v, q)\}}) \text{, otherwise}
    \end{align}

    Rule~3 ensures that only the monomials where every vertex is
    indexed by the same element in $[3]$ survive. The other rules
    ensure that any monomial $m$ indexed by $i\in[3]$ are mapped to
    $\sigma_i(m')$, where $m'$ is the same as $m$ but with $i$
    removed.
  \end{enumerate}

  The proof of correctness of these reductions is the same as the
  argument in Theorem~\ref{thm:supergraphs-full}. In addition, the
  condition that $u > v$ when $u$ is coloured $1$ and $v$ is coloured
  $k$ rules out one out of two automorphisms for $\overline{P_k}$ in
  part~2 and the condition that $u > v$ when $u$ is coloured $2$ and
  $v$ is coloured $k$ rules out one out of two automorphisms for
  $K_k - P_{k-1}$ in part~3.
\end{proof}

\begin{theorem}
  If $\ghom{K_k - P_{k-1}}$ can be computed by circuits of size
  $n^{f(k)}$, then there is an $O(n^{f(k)})$ time algorithm for
  induced subgraph isomorphism problem for $C_k$ on $n$-vertex graphs
  for odd $k\geq 5$.
  \label{thm:cycle-main-theorem-full}
\end{theorem}

\begin{theorem}
  \label{thm:AlgosForCycles-full}
  The following algorithms exist

  \begin{enumerate}
  \item An $O(n^\omega)$-time algorithm for induced subgraph
    isomorphism problem for $C_5$ in $n$-vertex graphs.
  \item An $O(n^{k-2})$-time combinatorial algorithm for induced
    subgraph isomorphism problem for $C_k$ in $n$-vertex graphs, where
    $k\geq 5$ is odd.
  \item An $O(n^{k-2})$-time deterministic combinatorial algorithm for
    computing the parity of the number of induced subgraphs isomorphic
    to $C_k$ in $n$-vertex graphs, where $k\geq 5$ is odd.
  \end{enumerate}
  \label{thm:cycle-algos-full}
\end{theorem}
\begin{proof}
  \begin{enumerate}
  \item We describe how to compute $\ghom{K_5 - P_4}$ using
  arithmetic circuits of size $O(n^\omega)$. We start by defining the
  following matrices.
    
    \begin{figure}
      \begin{center}
      \begin{tikzpicture}
        \node (1) at (0, 0) {1};
        \node (2) at (2, 0) {2};
        \node (5) at (3, -1) {5};
        \node (3) at (1, -2) {3};
        \node (4) at (-1, -1) {4};
        
        \draw (1) -- (2);
        \draw (1) -- (3);
        \draw (2) -- (3);

        \draw (1) -- (4);
        \draw (3) -- (4);
        \draw (2) -- (5);
        \draw (3) -- (5);
      \end{tikzpicture}
      \end{center}
      \caption{A labelled $K_5 - P_4$}
      \label{fig:k5-p4-full}
    \end{figure}
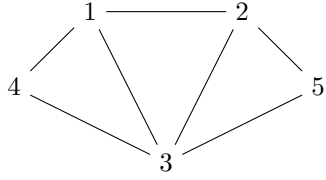

    \begin{align*}
      A^{n\times n}_{i, j} &= x_{\{(i, 1), (j, 3)\}}, i\neq j\\
      E^{n\times n}_{i, j} &= x_{\{(i, 3), (j, 2)\}}, i\neq j\\
      F^{n\times n}_{i, j} &= x_{\{i, j\}}, i\neq j\\
      B^{n\times n}_{i, i} &= y_i z_{3, i}\\
      C^{n\times n}_{i, i} &= y_i z_{4, i}\\
      D^{n\times n}_{i, i} &= y_i z_{5, i}
    \end{align*}

    Consider the labelled $K_5 - P_4$ in
    Figure~\ref{fig:k5-p4-full}. Compute the matrix products $FCF$, $FDF$,
    and $ABE$. Compute the polynomial
    $\sum_{i, j\in [n], i\neq j} z_{1, i} z_{2, j} y_i y_j x_{\{i,
      j\}} {(ABE)}_{i, j}$ and replace $x_{\{(i, 1), (j, 3)\}}$ with
    ${(FCF)}_{i, j}$ and replace $x_{\{(i, 3), (j, 2)\}}$ with
    ${(FDF)}_{i, j}$. It is easy to see that the resulting polynomial
    is $\ghom{K_5 - P_4}$ for this labelled $K_5 - P_4$ and the
    circuit has size $O(n^\omega)$.
  \item $\tw{K_k - P_{k-1}} = k-3$.
  \item The proof is similar to the proof of Part~4 of
    Theorem~\ref{thm:path-algos-full}.
  \end{enumerate}
\end{proof}

We remark that by computing homomorphism polynomials for $K_k -
P_{k-1}$ for $k = 7, 9$ using small-size circuits, we can obtain an
$O(n^{2\omega})$ time algorithm for induced subgraph isomorphism for
$C_7$ and an $O(n^{\omega(3,3,2)})$ time algorithm for induced
subgraph isomorphism for $C_9$. These algorithms are faster than the
corresponding algorithms for $k$-cliques.

\section{Algorithms for almost all induced patterns}
In this section, we prove a result that is similar in spirit to
Theorem~{\ref{thm:np-universal-full}} in \cite{NP85} which states that the
time complexity of induced subgraph isomorphism problem for $K_k$
upper bounds that of any $k$-vertex pattern graph. We show that the
circuit complexity of $\ghom{K_k - e}$ upper bounds the time
complexity of the induced subgraph isomorphism problem for all
$k$-vertex pattern graphs $H$ except $K_k$ and $I_k$. The algorithms
obtained from this statement can be obtained from known
results. However, we believe that restating these upper-bounds in
terms of circuits for $K_k - e$ homomorphism polynomials may give new
insights to improve these algorithms.

The key idea is that an efficient construction of homomorphism
polynomial for $K_k - e$ enables efficient construction of
homomorphism polynomials for all smaller graphs. First, we prove the
following technical result.

\begin{proposition}
If $\sub{H} \redto f$ and $f$ is a graph pattern polynomial family
with uniform $s(n)$-size circuits, then $\ghom H$ has uniform
$O(s(n))$-size circuits.
\end{proposition}
\begin{proof}
We can assume w.l.o.g.\ that $H$ does not have isolated vertices. Let
$H$ have $k$ nodes and let $K_{n}^k$ be the complete $k$-partite graph
with $n$ nodes in each partition. The nodes of $K_{n}^{k}$ are of the
form $(i, \kappa)$, $1 \le i \le n$, $1 \le \kappa \le k$. Let
$\sigma$ be a family of substitutions realizing $\sub H \redto f$.
Consider $\sub {H,kn}$. We know that $\ml{\sigma_{m}(f_m)} = v_{[q]}
\sub {H,kn}$ for some $m = O(n)$ and $q = O(1)$.  Since $H$ does not
contain isolated vertices, there is a function $g$ that maps $V(H)$ to
$E(H)$ such that the image of $f(v)$ for any $v$ is an edge incident
on $v$. Now we define the substitution $\tau$ on the edge and vertex
variables:

\begin{align*}
  \tau(x_{\{(i, \kappa),(j, \mu)\}}) &=
   \begin{cases}
    Y_{i, \kappa, \{\kappa, \mu\}} Y_{j, \mu, \{\kappa, \mu\}} x_{\{i,j\}} &
    \text{if $i \not= j$}, \\ 0 & \text{if $i = j$ or
      $\{\kappa,\mu\}\not\in E(H)$},
   \end{cases}\\
   \tau(y_{(i,\kappa)}) &= \hat a_\kappa,
\end{align*}

where the variables $\hat a$ are fresh variables that we need
for book-keeping and we define:

\begin{align*}
  Y_{i, \kappa', \{\kappa, \mu\}} &= z_{\kappa, i} y_i &\text{if $g(\kappa') = \{\kappa, \mu\}$}\\
  Y_{i, \kappa', \{\kappa, \mu\}} &= 1 &\text{if $g(\kappa') \not= \{\kappa, \mu\}$}
\end{align*}

Every embedding of $H$ into $K_{n}^k$ such that each node of $H$ goes
into another part will contribute a term that is multilinear in the
$\hat a_\kappa$-variables in $\tau(N_{H,kn})$. The substitution also
ensures that the colours of edges correspond to edges in $H$ and
labels of adjacent vertices are different. It is easy to see that
these embeddings correspond to homomorphisms to $K_n$. We have proved
that the part of $\tau(\sub{H, kn})$ multilinear in $\hat a$ variables
is,

\begin{equation*}
   \hat a_{V(H)} \sum_{\phi: H \hommap K_n} \prod_{v\in V(H)} z_{v,
     \phi(v)} y_{\phi(v)} \prod_{e\in E(H)} x_{\phi(e)} = \hat a_{[k]}
   \ghom {H,n}.
\end{equation*}

Furthermore the part of $\tau(\sigma_{m}(f_m))$ multilinear in $\hat
a$ and $v_i$ variables is $v_{[q]} \hat a_{[k]}\ghom {H,n}$ since
every non-multilinear term stays non-multilinear under $\tau$.
Therefore, we get an exact computation for $\ghom {H,n}$ by
differentiating the circuit with respect to $v_1,\dots,v_q,\hat
a_1,\dots,\hat a_k$ once and then setting all variables $v_i$ for all
$i$ and $\hat a_1,\dots,\hat a_k$ to $0$.  Note that each
differentiation will increase the circuit size by a constant factor
and we differentiate a constant number of times. This operation is
linear-time in the size of the circuit.\footnote{Note that unlike in
  the Baur-Strassen theorem, we only compute \emph{one} derivative!}
\end{proof}

The above result can be interpreted in two different ways: (1)
Homomorphism polynomials are the best graph pattern polynomials or (2)
Efficient constructions for homomorphism polynomials can be obtained
by obtaining efficient constructions for \emph{any} pattern family $f$
such that $\sub{H}\redto f$.

\begin{lemma}
  Let $k>2$. If $H\neq K_k$ is a $k$-vertex graph, then
  $2\sub{H}\redto \ghom{K_k-e}$.
\end{lemma}
\begin{proof}
  The proof of this claim is similar to the proof of
  Theorem~\ref{thm:ind_graph_to_clique-full}. Let $M$ be the labelling of
  $K_k - e$ using $[k]$ such that vertices $1$ and $k$ are not
  adjacent. Let $L$ be a labelling of $H$ using $[k]$ such that $1$
  and $k$ are not adjacent. Therefore, the labelled graph $L$ is a
  subgraph of the labelled graph $M$. Let $q_1, \dotsc, q_\ell$ be
  the edges of $L$ and $q_{\ell+1},\dotsc,q_m$ be the non-edges of
  $L$. Let $S$ be the set of all labellings of $H$. For each
  labelling $L'$ in $S$, associate a permutation with $L'$ such that
  applying it to $L'$ yields $L$. Let $P$ be the set of all such
  permutations.
  
  We partition $P$ into $P_1$ and $P_2$ as follows: A permutation
  $\phi\in P_1$ if given a sequence of $k$ numbers, we can determine
  whether the sequence is consistent with $\phi$, i.e., the
  $i^\text{th}$ smallest element in the sequence is at position
  $\phi(i)$, without comparing the first and last elements in the
  sequence. Otherwise, $\phi\in P_2$. We start with the
  $\ghom{K_k - e}$ polynomial over the vertex set
  $[n]\times [k]\times P$ and apply the following substitution.

  \setcounter{equation}{0}
  \begin{align}
    \sigma_{H}(y_{(v, p, \phi)}) &= y_v\\
    \sigma_{H}(x_{\{(v_1, p_1, \phi), (v_2, p_2, \phi')\}}) &= 0, \text{ if $\phi\neq \phi'$}\\
    \sigma_{H}(x_{\{(v_1, p_1, \phi), (v_2, p_2, \phi)\}}) &= 0, \phi^{-1}(p_1) < \phi^{-1}(p_2) \wedge v_1 > v_2\\
    \sigma_{H}(x_{\{(v_1, p_1, \phi), (v_2, p_2, \phi)\}}) &= \begin{cases}
      x_{\{v_1, v_2\}}, & \{p_1, p_2\}\in E(L)\\
      1, &\{p_1, p_2\}\in E(M)\setminus E(L)\\
      0, &\text{otherwise}
    \end{cases}\\
    \sigma_{H}(z_{(1, (v, 1, \phi))}) &= \begin{cases}
      u_1, &\phi\in P_2\\
      2u_1, &\phi\in P_1
    \end{cases}\\
    \sigma_{H}(z_{(i, (v, i, \phi))}) &= u_i, i > 1\\
    \sigma_{H}(z_{(i, (v, j, \phi))}) &= u_i^2, i \neq j
  \end{align}

  First, we state some properties satisfied by the surviving
  monomials. Rule~1 ensures that all vertices have different
  labels. Rule~2 ensures that all variables in a surviving monomial
  are indexed by the same permutation. Rules~6 and 7 ensure that all
  vertices have different colours. Let $\tau = (1\ k)(2)\dotsm
  (k-1)$. Consider an arbitrary surviving monomial indexed by a
  permutation $\phi$. If $\phi\in P_1$, then Rule~3 ensures that the
  vertices of the monomial are consistent with $\phi$. Assume that
  the vertices are $(v_1, 1, \phi),\dotsc, (v_k, k, \phi)$ and they
  are not consistent with $\phi$. This is possible only if
  $\phi^{-1}(1) < \phi^{-1}(k)$ and $v_1 > v_k$ or $\phi^{-1}(1) >
  \phi^{-1}(k)$ and $v_1 < v_k$. Since $\phi\in P_1$, there exists
  an $i'$ such that $\phi^{-1}(1) < \phi^{-1}(i') < \phi^{-1}(k)$ or
  $\phi^{-1}(1) > \phi^{-1}(i') > \phi^{-1}(k)$. Therefore, we have
  that the vertices are inconsistent at either $\{1, i'\}$ or $\{i',
  k\}$, a contradiction. If $\phi\in P_2$, then Rule~3 ensures that
  the vertices are consistent with $\phi$ or $\tau\circ\phi$. To see
  this, observe that, by Rule~3, the inconsistency with $\phi$ can
  only occur $\{1, k\}$. This implies that the vertices are
  consistent with $\tau\circ\phi$ because ${(\tau\circ\phi)}^{-1}(1)
  = \phi^{-1}(k)$ and ${(\tau\circ\phi)}^{-1}(k) = \phi^{-1}(1)$
  removing the inconsistency at $\{1, k\}$ and for all other $i$, we
  have ${(\tau\circ\phi)}^{-1}(i) = \phi^{-1}(i)$ preserving
  consistency at all other points.

  Consider a labelled $H$, say $L'$, labelled using
  $v_1 < \dotsb < v_k$ with associated permutation $\phi$. Let
  $\psi : v_i \mapsto i$. Let $e_1, \dotsc, e_m$ be the edges and
  non-edges of $L'$ such that $e_i = \psi^{-1}(\phi^{-1}(q_i))$ for
  all $i$. We split the proof into two cases: If $\phi\in P_1$, the
  monomial
  $z_{(1, (v_{\phi^{-1}(1)}, 1, \phi))}\dotsm z_{(1,
    (v_{\phi^{-1}(k)}, k, \phi))}$ (A monomial in $\ghom{K_k - e}$
  is completely determined by the homomorphism variables and we will
  not specify the other variables for brevity) uniquely generates
  the monomial in $\sub{H}$ that corresponds to $L'$. If
  $\phi\in P_2$, then there are two cases to consider depending on
  whether the permutation $\tau$ is in $Aut(L)$ or not. If
  $\tau\not\in Aut(L)$, then the monomials
  $z_{(1, (v_{\phi^{-1}(1)}, 1, \phi))}\dotsm z_{(1,
    (v_{\phi^{-1}(k)}, k, \phi))}$ and
  $z_{(1, (v_{\phi^{-1}(1)}, 1, \tau\circ\phi))}\dotsm z_{(1,
    (v_{\phi^{-1}(k)}, k, \tau\circ\phi))}$ are the only two
  monomials that yield the required monomial. If $\tau\in Aut(L)$,
  then the monomials
  $z_{(1, (v_{\phi^{-1}(1)}, 1, \phi))}\dotsm z_{(1,
    (v_{\phi^{-1}(k)}, k, \phi))}$ and
  $z_{(1, (v_{\phi^{-1}(k)}, 1, \phi))}\dotsm z_{(1,
    (v_{\phi^{-1}(1)}, k, \phi))}$ are the only two monomials that
  yield the required monomial.
\end{proof}

The above lemma shows that, as expected, the polynomial $\ghom{K_k -
  e}$ is strong enough to compute every other graph homomorphism
except that of $K_k$. This allows us to parameterize many existing
results in terms of the size of the arithmetic circuits computing
$\ghom{K_k - e}$.

\begin{theorem}
  If there are uniform $O(n^{s(k)})$ size circuits for $\ghom{K_k -
    e}$, then the number of subgraph isomorphisms for any $k$-vertex
  $H\neq K_k$ can be computed in $O(n^{s(k)})$ time on $n$-vertex
  graphs.
\end{theorem}
\begin{proof}
  For all $k$-vertex $H\neq K_k$, we have $2\sub{H}\redto \ghom{K_k -
    e}$. For all $H$ on less than $k$ vertices, we have
  $\sub{H}\redto\ind{K_k}\redto\ghom{K_k - e}$. Therefore, for all
  graphs $H\neq K_k$ on at most $k$ vertices, we can construct
  $O(n^{s(k)})$ size circuits that compute $2\ghom{H}$. We know that
  the number of subgraph isomorphisms for $H$ can be expressed as a
  linear combination of the number of homomorphisms for $H$ and the
  number of homomorphisms for graphs on less than $k$ vertices
  \cite{CDM17}.
\end{proof}

\begin{theorem}
  \label{thm:Kk-e-full}
  If there are uniform $O(n^{s(k)})$ size circuits for $\ghom{K_k -
    e}$, then the induced subgraph isomorphism problem for all
  $k$-vertex pattern graphs except $K_k$ and $I_k$ have an
  $O(n^{s(k)})$ time algorithm.
\end{theorem}
\begin{proof}
  We will show how to decide induced subgraph isomorphism for $H\neq
  K_k$ in $O(n^{s(k)})$ time. Now, choose a prime $p$ such that $p$
  divides the number of occurences of $H$ in $K_k$. The number of
  induced subgraph isomorphisms modulo $p$ for $H$ can be expressed as
  a linear combination of the number of subgraph isomorphisms modulo
  $p$ of $k$-vertex graphs except $K_k$ and can be computed in
  $O(n^{s(k)})$ time. It is known that the induced subgraph
  isomorphism problem for $H$ is randomly reducible to this problem
  \cite{WWWY15}.
\end{proof}

\begin{theorem}
  \label{thm:Kk-e-counting-full}
  If there are uniform $O(n^{s(k)})$ size circuits for $\ghom{K_k -
    e}$ and if there is an $O(t(n))$ time algorithm for counting the
  number of induced subgraph isomorphisms for a $k$-vertex pattern
  $H$, then the number of induced subgraph isomorphisms for all
  $k$-vertex patterns can be computed in $O(n^{s(k)} + t(n))$ time on
  $n$-vertex graphs.
\end{theorem}
\begin{proof}
 We know that $i_H = \sum_{H'\supg H} a_{H'} n_{H'}$, where all
 $a_{H'}\neq 0$, $i_H$ is the number of induced subgraph isomorphisms
 from $H$ to $G$ and $n_H$ is the number of subgraph isomorphisms from
 $H$ to $G$. Furthermore, we can compute $n_{H'}$ for all $H'\neq K_k$
 in $O(n^{s(k)})$ time. Therefore, if we can compute $i_H$ in $t(n)$
 time, we can compute $n_{K_k}$ in $O(n^{s(k)} + t(n))$ time.
\end{proof}

The following corollary follows by observing that $\tw{K_k - e} =
k-2$.

\begin{corollary}
  All $k$-vertex pattern graphs except $K_k$ and $I_k$ have an
  $O(n^{k-1})$ time combinatorial algorithm for deciding induced
  subgraph isomorphism on $n$-vertex graphs.
\end{corollary}

\begin{corollary}
  For $k\in\{4,5,6,7,8\}$, the induced subgraph isomorphism problem
  for any $k$-vertex pattern graph $H$ except $K_k$ and $I_k$ can be
  decided faster than currently known best clique algorithms.
\end{corollary}
\begin{proof}
  The polynomial family $\ghom{K_k - e}$ can be computed by uniform
  arithmetic circuits of size
  $O(n^{\omega(\lceil \frac{k-2}{2} \rceil, 1, \lfloor \frac{k-2}{2}
    \rfloor)})$ for all $k$. The construction is similar to the other
  constructions for homomorphism polynomials using fast matrix
  multiplication in this paper.
\end{proof}

\section{Reductions between patterns}

The following proposition is analogous to the obvious fact that the
complexity of the induced subgraph isomorphism problem is the same for
any pattern and its complement.

\begin{proposition}
  $\ind{H}\redto \ind{\overline{H}}$ for all graphs $H$.
\end{proposition}
\begin{proof}
  Use the substitution that maps $x_e$ to $1-x_e$ for any edge
  variable $x_e$ and maps any vertex variable to itself.
\end{proof}

It is known that $\naut{H} = 1$ for almost all graphs $H$. Therefore,
the following proposition can be interpreted as stating that the
homomorphism polynomial is harder than the subgraph isomorphism
polynomial for almost all pattern graphs $H$. This is used in
\cite{FominLRSR12} to obtain algorithms for subgraph isomorphism
problems.

\begin{proposition}
  $\naut{H}\sub{H} \redto \ghom{H}$ for all graphs $H$.
\end{proposition}
\begin{proof}
  Let $H$ be a $k$ vertex graph labelled using $[k]$. Use the
  substitution $\sigma(z_{a, v}) = u_a$ for all $a\in V(H), v\in V(G)$
  and $\sigma(w) = w$ for all the other variables $w$ in $\ghom{H}$
  over the vertex set $[n]$. We have $\naut{H}.u_{[k]}.\ml{\sub{H}} =
  \ml{\sigma(\ghom{H})} = \sigma(\ml{\ghom{H}})$. Consider an
  arbitrary automorphism $\phi$ of $H$. For every monomial $m =
  y_{v_1}\dotso y_{v_k}x_{e_1}\dotso x_{e_\ell}$ in $\sub{H}$, there
  are exactly $\naut{H}$ monomials $m_{\phi} = z_{(\phi(1),
    v_1)}\dotso z_{(\phi(k), v_k)}y_{v_1}\dotso y_{v_k}x_{e_1}\dotso
  x_{e_\ell}$ in $\ghom{H}$ that satisfy $\sigma(m_\phi) =
  u_{[k]}m$. This proves Properties~1 and 2 of the reduction. It is
  easy to see that the reduction satisfies the other properties too.
\end{proof}

Intuitively, the subgraph isomorphism problem should become harder
when the pattern graph becomes larger. However, it is not known
whether this is the case. Nevertheless, we can show this hardness
result holds for subgraph isomorphism polynomials for almost all
pattern graphs.

\begin{theorem}
  If $H\subg H'$, then $\naut{H}\sub{H} \redto \sub{H'}$.
  \label{thm:supergraphs-full}
\end{theorem}
\begin{proof}
  Let $|V(H)| = k$ and $|V(H')| = k + \ell$ for some $\ell \ge
  0$. Choose a labelling $L$ of the vertices of $H'$ such that the
  vertices of an $H$ in $H'$ are labelled $1,\dotsc, k$. Consider the
  polynomial $\sub{H'}$ over the vertex set $([n]\times [k]) \cup
  \{(n+i, k+i) : 1\le i\le \ell\}$. Substitute for the variables as
  follows:

  \setcounter{equation}{0}
  \begin{align}
    \sigma(y_{(i, p)}) &= \begin{cases}
      y_i u_p, &\text{ for all } i\in[n], p\in[k]\\
      u_p, &\text{ otherwise}
    \end{cases}\\
    \sigma(x_{\{(i_1, p_1), (i_2, p_2)\}})
    &=
    \begin{cases}
      x_{\{i_1, i_2\}} & \quad \text{if } \{p_1, p_2\}\in E(H)\\
      1 & \quad \text{if } \{p_1, p_2\}\in E(H')\setminus E(H)\\
      0 & \quad \text{otherwise}
    \end{cases}
  \end{align}

  We say that a monomial in $\sub{H'}$ \emph{survives} if the monomial
  does not become non-multilinear or 0 after the substitution. First,
  we will prove that all surviving monomials correspond to
  $H'$-subgraphs where the labels and colours of vertices are
  different and the colours of edges are the same as in the labelling
  $L$. Rule~1 ensures that the colours and labels of all vertices in
  the surviving monomials are different. Rule~2 ensures that there is
  a one-to-one correspondence between the edges $\{p_1, p_2\}$ in the
  labelling $L$ and the edge variables
  $x_{\{(i_1, p_1), (i_2, p_2)\}}$. To see this, observe that each
  monomial in $\sub{H'}$ has $|E(H')|$ edge variables. Since all
  vertices in a surviving monomial have different colours, all edges
  in the monomial must have different colours. Since any edge variable
  that has a colour not in the labelling $L$ is set to $0$, the
  colours of edges must be in one-to-one correspondence with the edges
  in the labelling $L$. This proves the all surviving monomials are of
  the form
  $y_{(u_1, 1)}\dotsm y_{(u_k, k)} (\prod_i y_{(n+i, k+i)}) x_{(e_1,
    q_1)} \dotsm x_{(e_m, q_m)} w$ for $u_1, \dotsc ,u_k\in [n]$,
  where $w$ is the product of edge variables with colour $\{p, q\}$
  such that $\{p, q\}$ is an edge in $H'$ but not in $H$ in the
  labelling $L$, $u_1, \dotsc, u_k$ are all different, and
  $q_1, \dotsc ,q_m$ are edges in $H$ in the labelling $L$. Note that
  the product $w$ is determined uniquely by $u_1,\dotsc,u_k$.

  We claim that for each monomial $y_{S} x_{T}$ in $\sub{H}$ over the
  vertex set $[n]$ there are $\naut{H}$ monomials $y_{S} x_{T}
  u_{[k]}$ in $\sigma(\sub{H'})$. Consider an arbitrary monomial $y_S
  x_T = y_{v_1}\dotsm y_{v_k} x_{e_1}\dotsm x_{e_m}$ in $\sub{H}$
  where $m = |E(H)|$. The monomials in $\sub{H'}$ that yield $y_S x_T
  u_{[k+\ell]}$ after the substitution are exactly the monomials $y_{(w_1,
    1)}\dotsm y_{(w_k, k)} (\prod_i y_{(n+i, k+i)}) x_{(e_1', q_1)}
  \dotsm x_{(e_m', q_m)} w$ where $w$ is the product of edge variables
  with colour $\{p, q\}$ such that $\{p, q\}$ is an edge in $H'$ but
  not in $H$ in the labelling $L$, $\{w_1, \dotsc, w_k\} =
  \{v_1,\dotsc, v_k\}$, and $\{e_1, \dotsc ,e_m\} = \{e_1',\dotsc,
  e_m'\}$. But this monomial corresponds to the automorphism $\phi :
  v_i\mapsto w_i$. Since $w$ is uniquely determined given
  $w_1,\dotsc,w_k$, the number of such monomials is $\naut{H}$. Also,
  each surviving monomial yields a monomial in $\sub{H}$.

  Additionally, each non-multilinear term in the polynomial obtained
  after the substitution contains at least one vertex or other
  variable with degree more than one. This proves the theorem.
\end{proof}

The following theorem states that the induced subgraph isomorphism
polynomial is harder than the subgraph isomorphism polynomial for
almost all graphs.

\begin{theorem}
  \label{thm:inducedHarderThanNonInduced-full}
  $\naut{H}\sub{H} \redto \ind{H}$ for all graphs $H$.
\end{theorem}
\begin{proof}
  Observe that $\ind{H} = \sub{H} + \sum_{H' \psupg H}
  a_{H'}\sub{H'}$. Let $k$ be the number of vertices in $H$ and fix
  some labelling of $H$ using $[k]$. Now consider the polynomial
  $\ind{H}$ over the vertex set $[n] \times [k]$ and apply the
  following substitution.

  \setcounter{equation}{0}
  \begin{align}
    \sigma(y_{(i, p)}) &= y_iu_p\\
    \sigma(x_{\{(i_1, p_1), (i_2, p_2)\}}) &= 
    \begin{cases}
      x_{\{i_1, i_2\}} & \text{ if $\{p_1, p_2\}\in E(H)$}\\
      0 & \text{ otherwise}
    \end{cases}
  \end{align}

  Now observe that any monomial in $\sub{H'}$ for $H' \psupg H$ must
  vanish because it will have at least one more edge than $H$. By the
  same argument as in the proof of Theorem~\ref{thm:supergraphs-full}, we
  conclude that there are exactly $\naut{H}$ monomials in $\sub{H}$
  over $[n] \times [k]$ that yield the monomial $y_Sx_Tu_{[k]}$ after
  the substitution for any monomial $y_Sx_T$ in $\sub{H}$ over $[n]$.
\end{proof}

We now prove the analogue of Theorem~{\ref{thm:np-universal-full}} in
\cite{NP85} which states that $k$-clique is harder than any other
$k$-vertex pattern graph.

\begin{theorem}
  For any $k$-vertex graph $H$, $\ind{H}\redto \ind{K_k}$.
  \label{thm:ind_graph_to_clique-full}
\end{theorem}
\begin{proof}
  Fix a canonical labelling $L$ of the graph $H$ using $[k]$.  Let
  $q_1, \dotsc ,q_\ell$ be the edges in the canonical labelling $L$
  and let $q_{\ell+1}, \dotsc , q_m$ be the non-edges in $L$ where
  $\ell$ is the number of edges in $H$ and $m = {\binom{k}{2}}$.  Let
  $S$ be the set of distinct labellings of $H$ using $[k]$. Associate
  all labellings $L'\in S$ with a permutation $\phi$ such that
  applying $\phi$ to an $H$ labelled $L'$ yields an $H$ labelled
  $L$. Let $P$ be the set of all such permutations. For example, there
  are three distinct labellings for $P_3$: $L =$ 1 -- 2 -- 3, 1 -- 3
  -- 2, and 2 -- 1 -- 3 with associated permutations $(1)(2)(3)$,
  $(1)(2 3)$, and $(1 2)(3)$ (Note that the these permutations are not
  unique if the graph has non-trivial automorphisms).  Apply the
  following substitution to $\ind{K_k}$ over the vertex set
  $[n]\times [k]\times P$:

  \setcounter{equation}{0}  
  \begin{align}
    \sigma(y_{(v, p, \phi)}) &= y_vu_p\\
    \sigma(x_{\{(v_1, p_1,\phi), (v_2, p_2,\phi')\}}) &= 0 \text{ if $\phi\neq \phi'$ or $p_1 = p_2$ or $v_1 = v_2$}\\
    \sigma(x_{\{(v_1, p_1,\phi), (v_2, p_2,\phi)\}}) &= 0
                                                       \text{ if $\phi^{-1}(p_1)< \phi^{-1}(p_2)$ and $v_1 > v_2$}\\
    \sigma(x_{\{(v_1, p_1,\phi), (v_2, p_2,\phi)\}}) &= 
    \begin{cases}
      x_{\{v_1, v_2\}} & \text{ if $\{p_1, p_2\}\in E(L)$}\\
      1 - x_{\{v_1, v_2\}} & \text{ if $\{p_1, p_2\}\not\in E(L)$}\\
    \end{cases}
  \end{align}

  The first two rules ensure that in any surviving monomial, the
  labels and colours of all vertices are different and all vertices
  are indexed by the same permutation.

  We can extend the correspondence between labellings of $H$ and
  permutations to arbitrary labellings (as opposed to labellings using
  $[k]$). Given a labelling of $H$ using $v_1 < \dotsb < v_k$, we can
  obtain a labelling $L'$ of $H$ using $[k]$ by replacing each $v_i$
  by $i$ for all $i$. The permutation associated with the labelling
  $M$ is the same as the permutation associated with labelling $L'$.

  Consider an arbitrary labelling $M$ of $H$ using
  $v_1 < \dotsb < v_k$ where each $v_i\in [n]$. Let $L'\in S$ be the
  labelling corresponding to the labelling $M$ such that
  $\psi : v_i \mapsto i$ is the permutation that maps $M$ to $L'$. Let
  $\phi\in P$ be the permutation associated with $L'$. For
  convenience, we denote the edges and non-edges of $M$ by
  $e_1,\dotsc,e_m$ such that $e_i = \psi^{-1}(\phi^{-1}(q_i))$ for all
  $i$. We will prove that for the term
  $t = y_{v_1}\dotsm y_{v_k}x_{e_1}\dotsm x_{e_\ell}(1 -
  x_{e_{\ell+1}})\dotsm (1 - x_{e_m})$ in $I(H)$ that encodes $M$,
  there is a unique monomial $s$ in $\ind{K_k}$ such that
  $\sigma(s) = u_{[k]}t$.  The monomial
  $s = y_{(v_1, \phi(1), \phi)}\dotsm y_{(v_k, \phi(k), \phi)}x_{(e_1,
    q_1, \phi)}\dotsm x_{(e_m, q_m, \phi)}$. First of all, we have to
  prove that given that $v_i$ has colour $\phi(i)$, the edges are
  coloured such that $e_i$ gets colour $q_i$. Start with an arbitrary
  $q_i = (j, k)$. Then,
  $e_i = \psi^{-1}((\phi^{-1}(j), \phi^{-1}(k))) = (v_{\phi^{-1}(j)},
  v_{\phi^{-1}(k)})$ which has colour $(j, k)$ as required. Also, we
  have $\sigma(s)\neq 0$ because if
  $\phi^{-1}(\phi(i)) = i < j = \phi^{-1}(\phi(j))$, then $v_i <
  v_j$. Given that $\sigma(s)\neq 0$, it is easy to see that
  $\sigma(s) = u_{[k]}t$ by applying rules 1 and 4.

  Given an arbitrary surviving monomial $r = y_{(v_1, 1, \phi)}\dotsm
  y_{(v_k, k, \phi)}\allowbreak x_{(e_1, q_1, \phi)}\dotsm x_{(e_m,
    q_m, \phi)}$ in $\ind{K_k}$ such that $\sigma(r) = u_{[k]}w$ for
  some $w$, we claim that $w$ encodes a labelling $M$ of $H$ where the
  permutation associated with $M$ is $\phi$. It is easy to see that
  $w$ encodes some labelling of $H$. Observe that for $r$ to survive,
  the vertices $(v_i, i, \phi)$ for all $i$ has to be consistent with
  $\phi$, i.e., the vertex coloured $\phi(i)$ must be the
  $i^{\text{th}}$ smallest among all $v_j$s by Rule~3. By the
  definition of $\phi$, we have $\{i, j\}\in E(L')$ if and only if
  $\{\phi(i), \phi(j)\}\in E(L)$. By Rule~4, we also have if
  $\{\phi(i), \phi(j)\}\in E(L)$ then $x_{\{v_{\phi(i)},
    v_{\phi(j)}\}}$ appears in the term $w$ and otherwise $(1 -
  x_{\{v_{\phi(i)}, v_{\phi(j)}\}})$ appears in $w$. In other words,
  in the graph encoded by $w$, the $i^{\text{th}}$ smallest and
  $j^{\text{th}}$ smallest vertices are connected if and only if the
  $i^{\text{th}}$ smallest and $j^{\text{th}}$ smallest vertices are
  connected in $L'$. Therefore, the associated permutation is $\phi$
  as claimed. We can now prove that $u_{[k]}t$ is uniquely generated
  from $s$. Suppose for contradiction that the monomial $s' =
  y_{(v_1', 1, \phi')}\dotsm y_{(v_k', k, \phi')} x_{(e_1', q_1,
    \phi')}\dotsm x_{(e_m', q_m, \phi')}$ also satisfies $\sigma(s') =
  u_{[k]}t$. Then, it must be that $\{v_1',\dotsc, v_k'\} =
  \{v_1,\dotsc ,v_k\}$, $\{e_1,\dotsc, e_\ell\} = \{e_1',\dotsc
  ,e_\ell'\}$, and $\{e_{\ell+1},\dotsc, e_m\} = \{e_{\ell+1}',\dotsc
  ,e_m'\}$. We know that $\phi = \phi'$ because the permutation in the
  monomial must correspond to the labelling encoded by $t$. But, $\phi
  = \phi'$ implies $v_i' = v_i$ for all $i$ (Otherwise, the third rule
  ensures that at least one edge variable in $s'$ becomes $0$ under
  $\sigma$). But, if $v_i' = v_i$ for all $i$, then $e_j = e_j'$ for
  all $j$ contradicting $s\neq s'$.
  
  We have proved that $\ml{\sigma(\ind{K_k})} =
  u_{[k]}\ind{H}$. Observe that the polynomial obtained after the
  substitution cannot contain edge variables of degree more than one
  because of Rule~2. It is easy to see that the substitution satisfies
  the other properties.
\end{proof}

The theorem below shows that the induced subgraph isomorphism
polynomial for any graph containing a $k$-clique or $k$-independent
set is harder than the $k$-clique polynomial. An analogous hardness
result is known for algorithms, only when the pattern $H$ contains a
$k$-clique (or $k$-independent set) that is disjoint from all other
$k$-cliques (or $k$-independent sets) \cite{FKLLTCS15}.

\begin{theorem}
  If $H$ contains a $k$-clique or a $k$-independent set, then
  $\ind{K_k} \redto \ind{H}$.
  \label{thm:clique-hard-full}
\end{theorem}
\begin{proof}
  We will prove the statement when $H$ contains a $k$-clique. The
  other part follows because if $H$ contains a $k$-independent set,
  then the graph $\overline{H}$ contains a $k$-clique and $\ind{K_k}
  \redto \ind{\overline{H}} \redto \ind{H}$.

  Fix a labelling of $H$ where the vertices of a $k$-clique are
  labelled using $[k]$ and the remaining vertices are labelled
  $k+1,\dotsc,k+\ell$. Consider the polynomial $\ind{H}$ over the
  vertex set $([n] \times [k]) \cup \{(n+i, k+i) : 1\leq i \leq \ell\}$
  and apply the following substitution.

  \setcounter{equation}{0}
  \begin{align}
    \sigma(y_{(i, p)}) &= 
    \begin{cases}
      y_iu_p & \text{ if $i\in [n]$ and $p\in[k]$}\\
      u_p & \text{ otherwise}
    \end{cases}\\
    \sigma(x_{\{(i_1, p_1), (i_2, p_2)\}}) &= 
    \begin{cases}
      x_{\{i_1, i_2\}} & \text{ if $\{p_1, p_2\}\in E(K_k)$ and $p_1 < p_2$ and $i_1 < i_2$}\\
      1 & \text{ if $\{p_1, p_2\} \in E(H) \setminus E(K_k)$}\\
      0 & \text{ otherwise}
    \end{cases}
  \end{align}

  Consider a $k$-clique on the vertices $i_1,\dotsc,i_k\in[n]$ on an
  $n$-vertex graph where $i_1 < \dotsb < i_k$. The monomial in
  $\ind{K_k}$ corresponding to this clique is generated uniquely from
  the monomial $y_{(i_1, 1)} \dotso y_{(i_k, k)}$$\prod_i y_{(n+i,
    k+i)}x_{\{(i_1, 1), (i_2, 2)\}}$$\dotso x_{\{(i_{k-1}, k-1), (i_k,
    k)\}} w$ in $\ind{H}$, where $w$ is the product of all edge
  variables corresponding to edges in $H$ but not in $K_k$. Note that
  Rules~1 and 2 ensure that in any surviving monomial, the labels and
  colours of all vertices are distinct and the colours of the edges
  must be the same as $E(H)$. The product $w$ is determined by $i_1,
  \dotsc, i_k$. This proves that $\ml{\sigma(\ind{H})} =
  u_{[k+\ell]}\ml{\ind{K_k}}$. It is easy to verify that the
  substitution satisfies the other properties.
\end{proof}

Theorem~\ref{thm:clique-hard-full} is true with $\sub{H}$ or
$\ghom{H}$ instead of $\ind{H}$. In fact, the same proof works for
$\sub{H}$. For $\ghom{H}$, use the substitution in the proof of
Theorem~\ref{thm:clique-hard-full} along with $z_{a, (v, a)} = u_a$
and $z_{a, (v, b)} = u_a^2$ when $a\ne b$ for all homomorphism
variables.

\section{Discussion}

Since the subgraph isomorphism and homomorphism polynomials for
cliques have the same size complexity, there is no advantage to be
gained by using homomorphism polynomials instead of subgraph
isomorphism problem. How hard is it to obtain better circuits for
$\ghom{K_k}$? As the following proposition shows, improving the size
of $\ghom{K_3}$ implies improving matrix multiplication.

\begin{proposition}
If $\sub{K_3}$ (or $\ind{K_3}$ or $\ghom{K_3}$) has $O(n^\tau)$-size
circuits then the exponent of matrix multiplication $\omega \le
\tau$.
\end{proposition}
\begin{proof}
  Let $G$ be the complete tripartite graph $T_n$ on $3n$-vertices with
  partitions of size $n$.  The vertex set of $T_n$ is
  $[3] \times [n]$. Instead of substituting a $1$ for every edge in
  $T_n$, we substitute the variables $a_{i,j}$ for edges
  $\{(1,i),(2,j)\}$, $b_{i,j}$ for edges $\{(2,i),(3,j)\}$, and
  $c_{i,j}$ for edges $\{(3,i),(1,j)\}$. The resulting polynomial is:
  
  \begin{equation*}
   N' = \sum_{i = 1}^n \sum_{j = 1}^n \sum_{k = 1}^n y_{1,i} y_{2,j} y_{3,k} \cdot a_{i,j} b_{j, k} c_{k, i}
 \end{equation*}

 We subsitute $1$ for all vertex variables and obtain
 \begin{equation*}
   N'' = \sum_{i = 1}^n \sum_{j = 1}^n \sum_{k = 1}^n  a_{i,j} b_{j, k} c_{k, i}
 \end{equation*}

 $N''$ has $O(n^\tau)$-size circuits. It is well-known that
 $\omega \le \tau$ follows from this, see e.g.~\cite{DBLP:journals/toc/Blaser13}.
\end{proof}

It is interesting to know whether such connections exist for $k > 3$.

\bibliographystyle{plainurl}
\bibliography{subgraphs}

\appendix

\section{Omitted Proofs}

\begin{proof}{(Of Theorem~\ref{thm:test:1-full})}
  Suppose that the arithmetic circuit given as input computes a
  polynomial on $n$ variables. For each variable $x_i$ select
  $a\in \mathbb{Z}_{p}$ and $y\in \{y_1,\dotsc,y_k\}$ uniformly at
  random and substitute $ay$ for $x_i$. Evaluate the circuit over
  $\mathbb{Z}_{p}[y_1,\dotsc,y_k]/\langle
  y_{1}^{2},\dotsc,y_{k}^{2} \rangle$ and accept if and only if the
  result is non-zero. The correctness of this algorithm can be proved
  by induction on the number of multilinear terms in the polynomial
  that are non-zero modulo $p$.
\end{proof}

\begin{proof}{(Of Theorem~\ref{thm:diaz-full})}
  We will describe how to construct an arithmetic circuit of size
  $O(n^{t+1})$ for $\ghom{H}$ where $t = \tw{H}$. The construction
  mirrors the algorithm in Theorem~{3.1} in \cite{DiazST02}. We start
  with a nice tree decomposition $D$ of $H$. Each gate in the circuit
  will be labelled by some node (say $p$) in $D$ and a partial
  homomorphism $\phi : V(H) \mapsto [n]$. The label is $I_p(\phi)$.

  Let $p$ be a node in the tree decomposition $D$. Construct the
  circuit in a bottom-up fashion as follows:

  \begin{description}
  \item [$p$ is a start node with $X_p = \{ a \}$] Add $n$ input gates
    labelled $I_p(\{(a, v)\})$ with the constant $1$ as value for each
    $v\in[n]$.
  \item [$p$ is an introduce node] Let $q$ be the child of $p$ and
    $X_p - X_q = \{ a \}$. Add gates labelled $I_p(\phi \cup \{(a,
    v)\}) = I_q(\phi)$ for each $v\in[n]$. Since there are at most
    $O(n^{t+1})$ choices for $\phi \cup \{(a, v)\}$, there are at most
    $O(n^{t+1})$ gates.
  \item [$p$ is a join node] Let $q_1$ and $q_2$ be the children of
    $p$. Add gates labelled $I_p(\phi) = I_{q_1}(\phi)
    . I_{q_2}(\phi)$. Since there are at most $O(n^{t+1})$ choices for
    $\phi$, there are at most $O(n^{t+1})$ gates.
  \item [$p$ is a forget node] Let $q$ be the child of $p$ such that
    $X_q - X_p = \{ a \}$. Add gates $I_p(\phi) = \sum_{v\in[n]} z_{a,
    v} y_v x_{\{v, u_1\}} \dotsm x_{\{v, u_k\}} I_q(\phi \cup \{(a,
    v)\})$ where $\{v, u_i\}, 1\leq i\leq k$ are the images of the
    edges incident on $a$ in partial homomorphism $\phi \cup \{(a,
    v)\}$. Note that there are $O(n)$ gates corresponding to the tuple
    $(p, \phi)$. Since $p$ is a forget node, there are at most
    $O(n^t)$ such tuples and therefore at most $O(n^{t+1})$ gates.
  \end{description}
  
\end{proof}

\end{document}